\documentclass[10pt,journal,comsoc]{IEEEtran}
\usepackage{graphicx} 
\usepackage{graphics}
\usepackage{subfigure}
\usepackage{multirow}
\usepackage{textcomp}
\usepackage{array}
\usepackage{xcolor}
\usepackage[hidelinks]{hyperref}

\usepackage{enumitem}   
\usepackage{siunitx}    
\usepackage{booktabs}   
\usepackage{tabularx} 	
\usepackage{makecell}   
\newcolumntype{Y}{>{\centering\arraybackslash}X} 

\usepackage{tabularray}
\usepackage{eso-pic}

\usepackage{tikz}
\usetikzlibrary{shapes.geometric, arrows, positioning, calc}

\usepackage{amsmath,amssymb,amsfonts}
\usepackage{bm}

\usepackage{algorithm}
\usepackage{algorithmic}
\usepackage[algo2e,linesnumbered,ruled]{algorithm2e}

\usepackage{amsthm}
\theoremstyle{definition}

\newtheorem{property}{Property}

\newtheorem{thm}{Theorem}

\SetKwProg{Fn}{Function}{}{end}

\SetAlgoSkip{}

\newcolumntype{M}[1]{>{\centering\arraybackslash}m{#1}}

\DeclareSymbolFont{symbolsC}{U}{txsyc}{m}{n}
\DeclareMathSymbol{\notniFromTxfonts}{\mathrel}{symbolsC}{61}

\title{Adaptive and Resilient Dual-Layer Resource Slicing for Hovering Aerial Backhaul Networks}

\author{Chuan-Chi~Lai,~\IEEEmembership{Member,~IEEE},
    and Jen-Hsiang~Li
    \IEEEcompsocitemizethanks{        
        \IEEEcompsocthanksitem{This research was supported by the National Science and Technology Council, Taiwan, under Grant Nos. NSTC 114-2221-E-194-062- and NSTC 115-2221-E-194-042-MY2. 
		This work was also partially supported by the Advanced Institute of Manufacturing with High-tech Innovations (AIM-HI) from the Featured Areas Research Center Program within the framework of the Higher Education Sprout Project by the Ministry of Education (MOE) in Taiwan. \emph{(Corresponding author: Chuan-Chi~Lai.)}}
        \IEEEcompsocthanksitem{Chuan-Chi~Lai is with the Department of Communications Engineering, National Chung Cheng University, Minxiong Township, Chiayi County 621301, Taiwan, and also with the Advanced Institute of Manufacturing with High-tech Innovations (AIM-HI), National Chung Cheng University, Minxiong Township, Chiayi County 621301, Taiwan (e-mail: chuanclai@ccu.edu.tw).}
		\IEEEcompsocthanksitem{Jen-Hsiang~Li is with the Department of Information Engineering and Computer Science, Feng Chia University, Taichung 407102, Taiwan.} 
    }
}

\IEEEtitleabstractindextext{%
\begin{abstract}
This paper investigates adaptive and resilient dual-layer resource slicing in hovering aerial agent (HAA)-assisted backhaul networks for heterogeneous 5G/6G services, including enhanced mobile broadband (eMBB), ultra-reliable and low-latency communications (URLLC), and massive machine-type communications (mMTC). To address the complex coupling of this dual-layer architecture in non-stationary environments, we propose the resilient adaptive priority orchestration enhanced twin delayed deep deterministic policy gradient (RAPO-TD3) framework. We introduce a novel double soft-max projection mechanism to map the continuous action space into physically feasible bandwidth distributions, ensuring strict constraint adherence. Additionally, a resilient adaptive priority orchestration (RAPO) mechanism is embedded to safeguard mission-critical URLLC latency. Crucially, we establish a rigorous mathematical foundation proving that our framework ensures Lipschitz continuity and satisfies the Robbins-Monro conditions for stable asymptotic convergence. Extensive simulations under non-stationary traffic demonstrate that our RAPO-TD3 framework achieves superior performance relative to PPO, DDPG, and traditional solvers. Notably, via the RAPO mechanism, our approach maintains URLLC satisfaction levels closely approaching theoretical optima even during 500\% demand surges. Furthermore, scalability evaluations indicate that sub-millisecond execution latencies strictly satisfy the 1 ms URLLC budget, demonstrating the performance efficacy of our proposed framework.
\end{abstract}

\begin{IEEEkeywords}
HAA-assisted backhaul, network slicing, RAPO-TD3, URLLC, priority optimization
\end{IEEEkeywords}}

\begin{document}

\AddToShipoutPictureBG*{
    \AtPageUpperLeft{%
        \raisebox{-1.3cm}[\height][\depth]{
            \makebox[\paperwidth][c]{
                \parbox{.9\textwidth}{
                    \centering
                    \footnotesize \copyright~2026 IEEE. Personal use of this material is permitted. However, permission to use this material for any other purposes must be obtained from the IEEE by sending a request to pubs-permissions@ieee.org.
                }%
            }%
        }%
    }%
}

\maketitle
\IEEEdisplaynontitleabstractindextext

%
\IEEEpeerreviewmaketitle

\section{Introduction}
\label{sec:intro}

\IEEEPARstart{T}{he} evolution of sixth-generation (6G) wireless networks demands global coverage, intelligent sensing, and programmable networking~\cite{10551400,8685766,10473184}. To realize this vision, 6G architectures must support heterogeneous services with conflicting requirements~\cite{10742112,10354514,10003191}. These include \textit{enhanced mobile broadband} (eMBB) for high-data-rate applications, \textit{ultra-reliable and low-latency communications} (URLLC) for mission-critical millisecond-level constraints, and \textit{massive machine-type communications} (mMTC) for hyper-dense Internet of Things connectivity~\cite{9729992,10537066}. However, terrestrial infrastructure deployment is frequently impeded by geographical barriers and protracted construction cycles, limiting capacity expansion during temporary hotspots or post-disaster recovery~\cite{10104142}. Similarly, traditional backhaul solutions lack the reconfigurability to adapt to rapid spatial-temporal traffic shifts. Consequently, \textit{hovering aerial agent} (HAA)-assisted backhaul networks have emerged as a promising architecture~\cite{10323247,s24061888}. Leveraging rapid deployment capabilities, HAA fleets form an agile three-dimensional topology, acting as both backhaul and access nodes to extend coverage and guarantee \textit{quality of service} (QoS) in dynamic environments~\cite{10186380,YU2023109533}.

Despite these advantages, HAA systems are inherently constrained by limited radio frequency bandwidth and finite battery capacity~\cite{10716774,9999295}. The primary challenge is providing differentiated eMBB, URLLC, and mMTC services within this restricted airspace. Coordinating these resources while maintaining strict isolation and service-specific satisfaction remains a critical optimization problem in 6G architectures.

Existing backhaul management predominantly relies on offline planning and static slicing, pre-partitioning spectral resources into fixed proportions~\cite{10186380,9459763}. Although effective under quasi-static loads, these approaches lack the agility to handle non-stationary traffic surges during large-scale events or emergency missions. Aligning with the 3GPP management and orchestration framework~\cite{3gpp_ts_28_533}, transitioning toward autonomous and adaptive slicing paradigms is essential to handle highly volatile demand. Without this adaptability, unmanaged eMBB spikes may starve mMTC connections or push URLLC latencies beyond survival thresholds, causing critical failures. 
To protect mission-critical links under resource scarcity, traditional frameworks often utilize hard preemption strategies that abruptly terminate lower-priority connections, inducing severe packet dropouts and system oscillations. Ultimately, these mechanisms lack the real-time granularity and structural differentiability needed to gracefully balance diverse QoS metrics within volatile aerial environments.

\textit{Deep reinforcement learning} (DRL) provides a potent framework for adaptive slicing through its intrinsic online interaction and autonomous policy tuning capabilities~\cite{8103164,10302364,10750921}. Nevertheless, conventional DRL models, such as \textit{deep deterministic policy gradient} (DDPG), struggle with high-dimensional continuous action spaces where core-layer slice ratios and HAA-layer sub-slice allocations are hierarchically coupled and jointly constrained~\cite{10026882}. 
During exploration, standard continuous-control algorithms frequently generate unconstrained outputs violating physical bandwidth simplex boundaries, which compromises resource feasibility and destabilizes policies. Furthermore, hard-threshold reward functions often cause sparse gradients and convergence instability. To address these limitations, this study proposes the \textit{resilient adaptive priority orchestration enhanced twin delayed deep deterministic policy gradient} (RAPO-TD3) framework. Building upon the baseline TD3 algorithm (which employs twin-Q networks to suppress overestimation bias, delayed actor updates for stability, and target policy smoothing to mitigate oscillations~\cite{8103164,10257231}), our approach integrates a \textit{double soft-max projection} technique with differentiable penalty shaping. This structural design ensures decisions strictly adhere to feasible resource regions while preserving continuous gradients for efficient learning.

This research deploys this intelligent architecture for dual-layer network slicing in HAA-assisted backhaul systems. By incorporating \textit{prioritized experience replay} (PER)~\cite{PER_2016} alongside our newly developed \textit{resilient adaptive priority orchestration} (RAPO) mechanism, the proposed solution enhances operational resilience against stochastic traffic bursts without brittle, non-differentiable hard preemption strategies.

\begin{itemize}
    \item \textit{Adaptive Dual-Layer Architecture and Differentiable Projection}: We formulate a hierarchical slicing model for aerial backhaul to maximize heterogeneous service satisfaction. A novel double soft-max projection transforms the unconstrained action space into feasible bandwidth allocations, guaranteeing simplex constraint adherence during training and execution without heuristic boundary rules.
    \item \textit{Resilient Adaptive Priority Orchestration Mechanism}: The core RAPO mechanism is embedded within a differentiable penalty-driven reward function utilizing adaptive smooth gating. This provides dense gradient propagation, enabling the agent to dynamically shift resource focus and resiliently safeguard URLLC reliability under extreme workloads.    
    \item \textit{Theoretical Convergence and Stability Guarantees}: We establish a rigorous mathematical foundation for algorithmic stability. We formally prove that the continuously differentiable design ensures action space boundedness and Lipschitz continuity, satisfying the Robbins-Monro conditions~\cite{Robbins1951} to mathematically guarantee asymptotic convergence to a stable local optimum without divergent oscillations.
    \item \textit{Performance and Scalability Verification}: Extensive simulations under non-stationary traffic demonstrate that RAPO-TD3 achieves superior performance relative to DDPG and PPO baselines. By integrating the proposed Double Soft-max Projection mechanism, the RAPO-TD3 framework enables URLLC satisfaction levels closely approaching theoretical optima during 500\% traffic surges, restricting steady-state constraint violation penalties to approximately 0.08 compared with 3.3 in unconstrained baselines, while strictly maintaining sub-millisecond execution latencies.
\end{itemize}

\section{Related Work} 
\label{sec:related_work} 

Resource management in aerial-assisted networks has evolved significantly with the integration of network slicing and advanced machine learning techniques. This section reviews relevant literature across three key domains: network slicing paradigms, aerial deployment optimization, and deep reinforcement learning frameworks.

\subsection{Network Slicing Paradigms} 

Network slicing provides the essential flexibility and scalability required to adapt to evolving service demands, facilitating rapid application deployment and efficient service scaling. However, the coordinated resource allocation across heterogeneous slices, encompassing eMBB for high throughput, URLLC for low latency, and mMTC for massive connectivity, remains a formidable challenge.

Extensive research has investigated resource management strategies tailored for these diverse slices. Early studies in~\cite{Xiao2018} formulated the resource slicing problem over licensed and unlicensed bands as a mathematical optimization model to maximize the aggregate utility of heterogeneous services.
To manage the coexistence of eMBB and URLLC services, the integration of slicing with physical layer technologies, such as \textit{non-orthogonal multiple access} (NOMA), was demonstrated in~\cite{9448942} as an essential approach to balancing high throughput with stringent latency constraints. Beyond fundamental allocation, network slicing has been conceptualized as a critical backbone for IoT connectivity in smart cities, with automated scheduling schemes proposed in~\cite{9232929} to satisfy diverse sensing task requirements.

Current research predominantly focuses on \textit{end-to-end} (E2E) resource coordination and service layer orchestration. For instance, cloud-native container technology has been explored to establish flexible slicing service chains between softwarized HAA nodes and terrestrial core networks~\cite{10186380}. In the context of aerial access and backhaul, slice-aware joint optimization of deployment and spectrum allocation has been shown to significantly reduce resource consumption, particularly when aerial nodes are limited~\cite{9975290}. Recent advancements have further extended slicing paradigms to emerging fields, such as \textit{integrated sensing and communication} (ISAC)~\cite{10473184} and large-scale IoT frameworks for 6G intelligent infrastructures~\cite{10551400}. Collectively, achieving fine-grained bandwidth and power allocation while ensuring QoS and security isolation remains a pivotal issue in the evolution of slicing technologies.

\subsection{HAA Deployment and Backhaul Resource Management} 

The proliferation of HAA technology has facilitated the development of three-dimensional network topologies that integrate aerial access with backhaul connectivity. While HAAs enable rapid deployment for emergency response or temporary events, they are governed by stringent physical constraints, including finite bandwidth, limited power, and battery endurance. Additionally, backhaul links in these environments are highly susceptible to environmental fluctuations and line-of-sight obstructions.

Initial research in this domain focused primarily on front-end access optimization. For instance, stochastic geometry analysis has been employed to reveal the sensitivity of uplink interference to mMTC success rates, emphasizing the critical balance between deployment density and power control~\cite{10104142}. Surveys on the intersection of 6G and artificial intelligence have further positioned HAAs as portable edge inference nodes, highlighting their utility for latency-sensitive services within the \textit{Internet of Vehicles} (IoV) and smart city infrastructures~\cite{s24061888}. Regarding the coordination of multiple aerial nodes, an adaptive and fair deployment approach was proposed in~\cite{9946428} to balance offload traffic in multi-HAA cellular networks, ensuring equitable resource distribution among heterogeneous ground users.

Recent trends have shifted toward the joint design of access and backhaul integration. A dual-scale approach was proposed in~\cite{10186380}, where HAA positions and backhaul capacities are optimized offline, while slice resources are dynamically allocated online using DRL to satisfy service-level requirements. In the realm of \textit{integrated access and backhaul} (IAB), cooperative parallel resource allocation mechanisms, such as CPReal, have been introduced to manage both bursty and non-bursty traffic profiles~\cite{YU2023109533}. Furthermore, issues regarding fairness in multi-operator architectures~\cite{10716774} and real-time trade-offs between eMBB and URLLC traffic in multi-hop relay networks~\cite{10537066} have been investigated. Federated learning has also been explored to facilitate decentralized slicing decisions under multi-operator sharing and open \textit{radio access network} (Open RAN) frameworks~\cite{9999295}.

\subsection{DRL-based Resource Management} 

Traditional optimization and heuristic algorithms often struggle with the high-dimensional and non-stationary environments of 5G/6G networks. Consequently, DRL has been widely adopted for wireless resource management due to its ability to learn optimal policies through trial-and-error interactions. Integrating \textit{deep neural networks} (DNNs) with reinforcement learning has enabled systems to transcend the dimensionality constraints inherent in traditional methods~\cite{8103164}, with the trade-off between exploration and exploitation serving as a critical factor for convergence quality~\cite{LADOSZ20221}.

In network slicing, predictive-assisted DRL has demonstrated potential by achieving significant latency reductions through real-time optimization of power and user access~\cite{10302364}. For heterogeneous networks with \textit{mobile edge computing} (MEC), \textit{multi-agent} TD3 (MATD3) algorithms have been successfully implemented to maximize spectral efficiency~\cite{10750921}. To enhance service orchestration, federated DRL frameworks have been proposed to coordinate multiple xAPPs in Open RAN~\cite{s22083031}, while dual-granularity RAN slicing strategies have been developed to concurrently maximize QoS and resource utilization across the entire slice lifecycle~\cite{9459763}.

Advanced architectures, such as hierarchical DRL, have been introduced to coordinate control between the RAN and the core layer, effectively reducing backhaul signaling overhead~\cite{10026882}. In HAA-assisted scenarios, the synergy between DRL and energy prediction models has proven effective for joint trajectory and bandwidth scheduling~\cite{10186380}. 
Recent literature in 2026 has further emphasized the necessity of handling complex service coexistence and dynamic adaptations. For instance, advanced soft actor-critic architectures have been proposed to manage the stringent coexistence of URLLC traffic with distributed learning services through intelligent device selection~\cite{Ganjalizadeh2026}. Similarly, enhanced TD3 algorithms featuring dynamic normalization and adaptive weights have been successfully applied to balance delay and energy consumption in aerial collaborative computing~\cite{Song2026}. These latest developments lay the foundation for the dual-layer continuous control framework proposed in this study.

\subsection{Summary of Research Gaps and Motivation} 

A comprehensive comparison of the proposed framework with existing literature is summarized in Table~\ref{tab:literature_comparison}. Despite the significant advancements reviewed in the previous subsections, including the very recent progress in DRL-driven service coexistence~\cite{Ganjalizadeh2026} and adaptive TD3 offloading~\cite{Song2026}, several critical research gaps remain unaddressed. 
As categorized in Table~\ref{tab:literature_comparison} under the Slicing Hierarchy column, most existing literature treats HAAs merely as extensions of terrestrial base stations, predominantly focusing on single-layer resource division between the \textit{base station} (BS) and HAAs. Such simplified models often overlook the two-tier coupling mechanism between core-layer slice proportions and HAA-layer sub-slice allocations, failing to capture the complex interdependencies within a multi-layer backhaul network.

\begin{table*}[!t]
    \footnotesize
    \caption{Comparison of Related Literature}
    \label{tab:literature_comparison}
    \centering
    \begin{tblr}{
        width=\textwidth,
        colspec={Q[25mm]|X|X|X|X},
        row{1}   = {font=\bfseries,rowsep=2pt},  
        row{2-Z} = {rowsep=0pt},                 
        rows     = {valign=m},                   
        hline{1,2,Z} = {-}{},
    }
    Paper / Perspective & Slicing Hierarchy & Resource Type & QoS Metrics & Algorithm \\
    \cite{Xiao2018} (2018) & Single-layer (Core) & Bandwidth (Licensed / Unlicensed) & Aggregate Utility & Game Theory / Distributed Algo. \\
    \cite{YU2023109533} (2023) & IAB Integrated & Bandwidth, Backhaul Slot & URLLC Latency & MATD3 \\
    \cite{10257231} (2023) & Single-layer (HAA) & Bandwidth, HAA Position & SLA Satisfaction (Rate) & DQN \\
    \cite{10323247} (2024) & Hybrid Air-Ground & Bandwidth (Access + Backhaul) & eMBB Throughput and Coverage & Genetic-based Deployment \\
    \cite{Ganjalizadeh2026} (2026) & Single-layer (Edge) & Device Selection, Bandwidth & URLLC Latency, Convergence Time & BSAC \\
    \cite{Song2026} (2026) & Three-layer (Collaborative) & Computing, Energy & Delay, Energy Consumption & DN-TD3 \\
    Proposed Scheme & Dual-layer: Core + HAA Sub-slice & Bandwidth (Access + Backhaul) & Rate, Latency, and Success Rate & RAPO-TD3  
    \end{tblr}
\end{table*}

To bridge these gaps, this paper proposes a dual-layer framework that first performs slicing at the core layer and subsequently subdivides these resources at the HAA layer for mMTC, URLLC, and eMBB services. Our work distinguishes itself from existing studies in the following key aspects:
\begin{itemize}
    \item \textit{Algorithmic Robustness}: We combine the stability of the TD3 algorithm with the sample efficiency of PER. This combination is specifically tailored to handle the stringent latency constraints of URLLC in dynamic and non-stationary HAA backhaul contexts.
    \item \textit{Constraint-Aware Design}: Unlike existing DRL-based slicing literature that often bypasses or simplifies action-space constraints, we introduce a double soft-max projection mechanism. This provides a theoretically grounded solution to the action-space constraint problem, ensuring that resource allocations remain physically feasible while maintaining continuous gradients for efficient learning.
    \item \textit{Comprehensive QoS Support}: In contrast to works focused on a single performance metric, our framework simultaneously optimizes rates, latency, and success rates for heterogeneous services. This multi-objective optimization is reflected in the QoS Metrics column of Table~\ref{tab:literature_comparison}, demonstrating a more holistic approach to 6G service requirements.
    \item \textit{Adaptive Priority Execution}: To overcome the limitations of offline static slicing, we introduce the RAPO mechanism. This provides an adaptive framework to gracefully scale down best-effort services and protect mission-critical links under severe resource scarcity, an operational flexibility largely absent in prior studies.
\end{itemize}

By addressing the two-tier coupling and the action-space constraint simultaneously, the proposed RAPO-TD3 framework armed with RAPO offers a robust, continuous-control solution for the next generation of aerial backhaul networks.

\section{System Model and Problem Formulation}
\label{sec:system_model_problem_formulation}

As depicted in Fig.~\ref{fig:system_model}, the studied system operates under a stratified dual-layer slicing framework consisting of a centralized \textit{macro slicing orchestrator} (MSO) and an aerial radio access network. The MSO partitions the total bandwidth among HAAs on a macroscopic time scale based on aggregate demands, while each HAA subsequently distributes its allocated share among three local service slices on a sub-millisecond execution scale. This hierarchical structure mathematically decouples the spatial-temporal optimization scales, enabling fine-grained, real-time adaptation to non-stationary traffic while mitigating the state-action dimensionality explosion inherent in flat centralized frameworks.

\subsection{Network Architecture}
We consider a 6G-oriented HAA-assisted downlink backhaul network comprising a macro base station hosting the MSO and a set of $N$ HAAs indexed by $n = 1, 2, \dots, N$. To support diverse QoS requirements, the network is partitioned into three primary network slices: eMBB, URLLC, and mMTC. 
To mitigate terrestrial non-line-of-sight propagation blockages from urban terrain obstructions, HAAs are deployed at strategic altitudes to maintain clear line-of-sight links with the macro base station. This architecture adopts a \textit{hierarchical control paradigm}, operating the macro base station strictly as a centralized macro orchestrator to avoid high-frequency signal penetration losses and minimize core network computational overhead. Furthermore, by employing infrastructure-grade tethered HAAs, the system not only secures continuous power supply through physical tethers, thereby eliminating transient battery depletion limitations, but also gains high structural stability against aerial wind disturbances. This robust physical anchoring guarantees stable line-of-sight backhaul links, justifying their deployment as reliable quasi-stationary relays.

During backhaul service delivery, the HAAs maintain a quasi-stationary profile, treating their spatial coordinates as fixed boundary parameters. Building on this stable operational assumption, aerodynamic flight paths, three-dimensional trajectory optimization, and altitude adjustments are decoupled from the resource allocation problem. Furthermore, propulsion power consumption and flight energy limitations are likewise isolated from the operational action space. This structural configuration ensures that the dual-layer slicing framework can focus entirely on multi-service queueing dynamics and spectrum efficiency during the localized execution window.

\begin{figure}[!t]
	\centering
	\includegraphics[width=\linewidth]{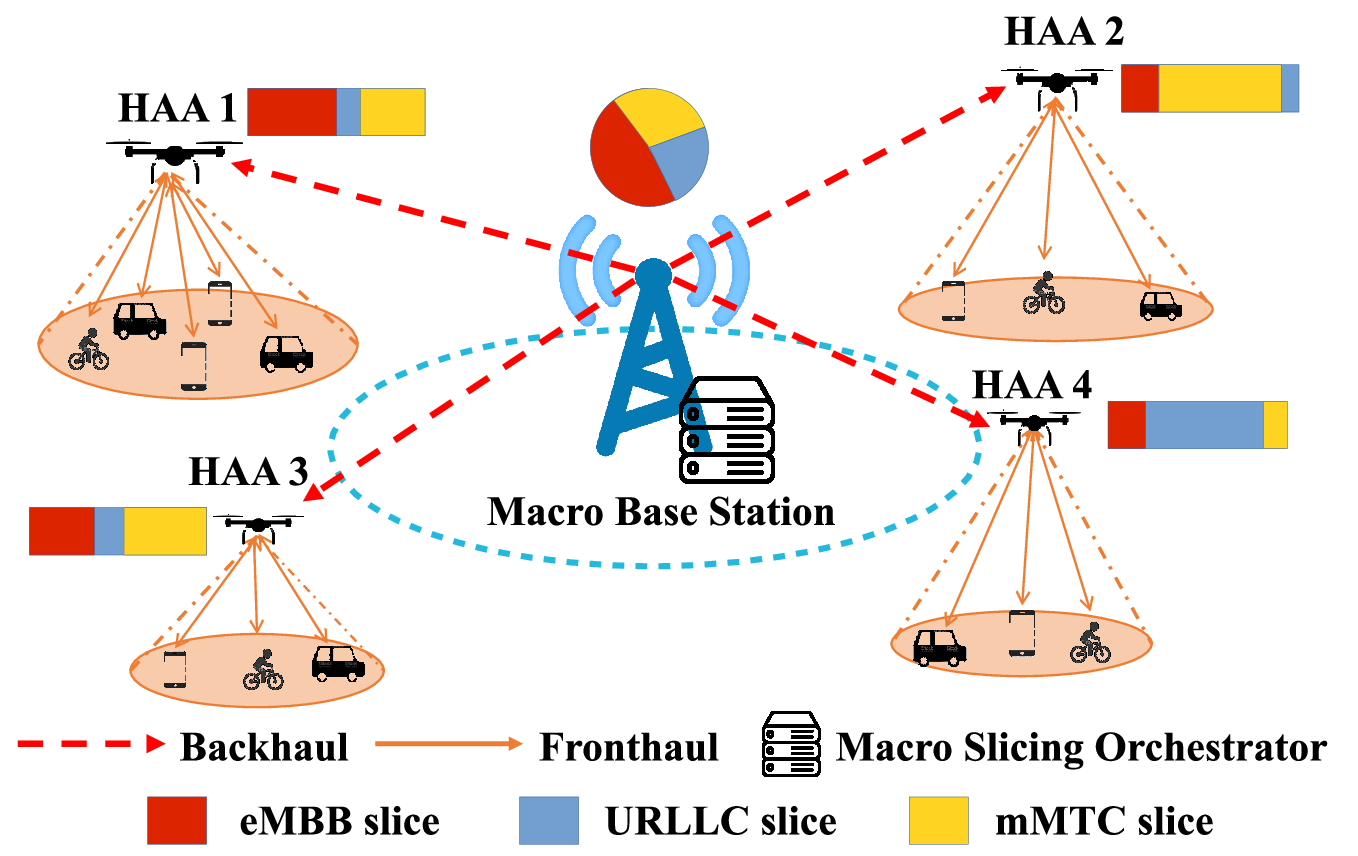} 
	\caption{HAA-assisted backhaul network with dual-layer slicing architecture.}
	\label{fig:system_model}
\end{figure}

\subsection{Heterogeneous Air-to-Ground Channel Model with Spatial Abstraction}
To characterize the spatial heterogeneity of the aerial transport stratum, we consider a network layout where $N$ HAAs are deployed at a constant operational altitude $H_{\text{haa}}$. The horizontal distance between the centralized macro base station and the $n$-th HAA is denoted by $r_n$, yielding a deterministic three-dimensional propagation distance of $d_n = \sqrt{H_{\text{haa}}^2 + r_n^2}$. Following the standard specifications for dense urban environments~\cite{al2014optimal}, the \textit{line-of-sight} (LoS) probability between the centralized orchestrator and the $n$-th HAA is dynamically governed by the spatial elevation angle $\varphi_n = \arctan(H_{\text{haa}} / r_n) \times \frac{180}{\pi}$, expressed as:
\begin{equation}
P_{\text{LoS}}(\varphi_n) = \frac{1}{1 + a_{\text{env}} \exp \left[ -b_{\text{env}} \left( \varphi_n - a_{\text{env}} \right) \right]},
\end{equation}
where $a_{\text{env}}$ and $b_{\text{env}}$ represent environment-specific constants that dictate the steepness and the offset of the propagation probability curve corresponding to the local urban building densities~\cite{al2014optimal}. Consequently, the macro-scale geometric path loss $\overline{PL}_n$ (in linear scale) is formulated as the expectation over LoS and \textit{non-line-of-sight} (NLoS) states:
\begin{equation}
\overline{PL}_n = P_{\text{LoS}}(\varphi_n) 10^{\frac{PL_{\text{LoS}}(d_n)}{10}} + \left[ 1 - P_{\text{LoS}}(\varphi_n) \right] 10^{\frac{PL_{\text{NLoS}}(d_n)}{10}},
\end{equation}
where the components $PL_{\text{LoS}}(d_n)$ and $PL_{\text{NLoS}}(d_n)$ incorporate free-space path loss alongside shadowing mid-scale attenuations at a designated carrier frequency $f_c$:
\begin{align}
PL_{\text{LoS}}(d_n) &= 20 \log_{10}(d_n) + 20 \log_{10}(f_c) + C_{\text{FSPL}} + \eta_{\text{LoS}}, \\
PL_{\text{NLoS}}(d_n) &= 20 \log_{10}(d_n) + 20 \log_{10}(f_c) + C_{\text{FSPL}} + \eta_{\text{NLoS}}
\end{align}
with $C_{\text{FSPL}}$ denoting the free-space reference constant derived from the speed of light. The parameters $\eta_{\text{LoS}}$ and $\eta_{\text{NLoS}}$ signify the average additional attenuation factors for LoS and NLoS links, respectively, which explicitly capture the structural penetration and obstacle reflection losses inherent to the 3D urban geometry~\cite{al2014optimal}.

To guarantee compliance with reinforcement learning convergence boundaries without undermining physical rigor, the instantaneous channel gain $h_n(t)$ at a discrete time slot $t$ governing the backhaul transport capability is modeled as a composite framework merging deterministic path loss and fast-varying small-scale fading:
\begin{equation}
h_n(t) = \Omega_n(t) \cdot \zeta_n,
\end{equation}
where $t$ indexes the discrete operational intervals, $\Omega_n(t)$ models the instantaneous small-scale multi-path fading profile governed by a Rayleigh distribution~\cite{khawaja2019survey}, and $\zeta_n$ represents the normalized large-scale spatial structural parameter mapping the path-loss heterogeneity onto the unified agent observation space:
\begin{equation}
\zeta_n = \frac{\overline{PL}_n^{-1/2}}{\frac{1}{N}\sum_{j=1}^{N} \overline{PL}_j^{-1/2}}.
\end{equation}

\subsection{Fronthaul Abstraction and Demand Modeling}
\label{sec:fronthaul_abstraction}
A major architectural feature of the proposed hierarchical management plane is the functional decoupling of lower-tier edge fronthaul access from upper-tier downlink backhaul transport provisioning. The system operates over a discrete-time horizon indexed by $t=1, \dots, T$, where each slot corresponds to a centralized macro control period $T_C$. At the beginning of each interval $t$, the MSO collects state information, executes global resource slicing orchestration, and disseminates decisions for localized enforcement. Crucially, the index $t$ denotes this macro-scale orchestration interval; once the sub-slice ratios $\mathbf{g}_i(t)$ are determined at the start of slot $t$, each HAA executes these assigned bandwidth proportions continuously across localized, sub-millisecond physical resource block scheduling loops throughout the entire duration of $T_C$. This structure strictly isolates the global optimization step from high-frequency fronthaul execution, which mathematically resolves time-scale separation without conflating the indices. 

We assume that microscopic user equipment association, intra-cell proportional fair scheduling, and physical-layer radio resource block tiling are executed autonomously at the network edge by decentralized base station components embedded within each HAA. From the operational perspective of the centralized MSO, these dynamic edge radio dynamics are systematically compressed and filtered. The decentralized edge mechanisms map the instantaneous local user traffic states into an aggregated, multi-service exogenous downlink traffic demand profile vector for each HAA $n$ at control period $t$, denoted as:
\begin{equation}
\boldsymbol{\Delta}_n(t) = \left[ \Delta_{n,e}(t), \Delta_{n,u}(t), \Delta_{n,m}(t) \right]^\top,
\end{equation}
where $\Delta_{n,e}(t)$, $\Delta_{n,u}(t)$, and $\Delta_{n,m}(t)$ represent the aggregated downlink traffic demands for eMBB, URLLC, and mMTC, respectively. Simultaneously, each HAA reports its instantaneous backhaul channel measurement $|h_n(t)|$ along with this aggregated demand vector $\boldsymbol{\Delta}_n(t)$ back to the MSO via the dedicated backhaul control channel.

To quantify the traffic pressure experienced by each HAA within every control period $T_C$ across the three heterogeneous service types, the downlink traffic demand is mathematically characterized by the total generated data volume expressed in bits. Specifically, the demand observed at the $n$-th HAA for service type $i \in \{e,u,m\}$ is parameterized as the summation of downlink user-plane packet arrivals routed from the core network:
\begin{equation}\label{eq:delta_nii}
\Delta_{n,i}(t) = \sum_{k \in \mathcal{K}_{n,i}} V_{k,i}(t),
\end{equation}
where $\mathcal{K}_{n,i}$ denotes the set of active user devices belonging to service type $i$ associated with HAA $n$, and $V_{k,i}(t)$ represents the discrete data volume generated for user $k$ during the current control interval. 
To guarantee compliance with physical causality, the traffic demand vector $\boldsymbol{\Delta}_n(t)$ observed by the centralized orchestrator at the exact start of control period $t$ represents the accumulated traffic volume that has arrived from the physical core network and populated the transmission buffers during the immediately preceding interval. Therefore, the centralized orchestrator operates strictly on materialized, deterministic queue states awaiting backhaul transport transmission rather than predicting future stochastic traffic bursts within slot $t$, ensuring complete mathematical causality in real-time execution.

To eliminate observation bias caused by structural spatial differences in the user population densities across distinct HAAs, the framework normalizes the absolute bit load into a dimensionless demand ratio:
\begin{equation}
\delta_{n,i}(t) = \frac{\Delta_{n,i}(t)}{\sum_{j \in \{e,u,m\}} \Delta_{n,j}(t)}, \quad i \in \{e,u,m\}.
\end{equation}
Here, $\delta_{n,i}(t)$ explicitly defines the fractional load contribution of service-type-$i$ demand observed at HAA $n$, satisfying the probability simplex constraint $\sum_{i \in \{e,u,m\}} \delta_{n,i}(t) = 1$. The complete vectorized demand state mapping is thus defined as $\boldsymbol{\delta}_n(t) = [\delta_{n,e}(t), \delta_{n,u}(t), \delta_{n,m}(t)]^\top$.

This aggregated vector follows non-stationary distribution boundaries parameterized by macro diurnal tidal trends and stochastic Poisson burst multipliers. By capturing the time-varying multi-service load at each HAA through this macro abstraction, the system accurately feeds the DRL controller. Isolating the centralized state-space from high-frequency fronthaul user mutations successfully bypasses the dimensionality curse, ensuring that the global core network management loop remains scalable, stable, and mathematically tractable.

\subsection{Dual-Layer Slice Resource Allocation Model}

As shown in Fig.~\ref{fig:dual_slicing_layer}, let $S_i(t)$ denote the macro-layer slice ratio allocated to service type $i \in \{e,u,m\}$ at control period $t$. The vector $\mathbf{S}(t) = [S_e(t), S_u(t), S_m(t)]$ must satisfy the following constraints:
\begin{align}
	& S_i(t) \geq 0, \quad \forall i \in \{e,u,m\}, \label{eq:core_slice_nonnegative}\\
	& \sum_{i \in \{e,u,m\}} S_i(t) = 1. \label{eq:core_slice_sum}
\end{align}
Here, constraint~\eqref{eq:core_slice_nonnegative} ensures that each slice ratio is non-negative, while constraint~\eqref{eq:core_slice_sum} guarantees that the total allocated slice ratios sum to unity, reflecting the complete partitioning of available resources at the macro layer.

\begin{figure}[!t]
	\centering
	\includegraphics[width=.3\textwidth]{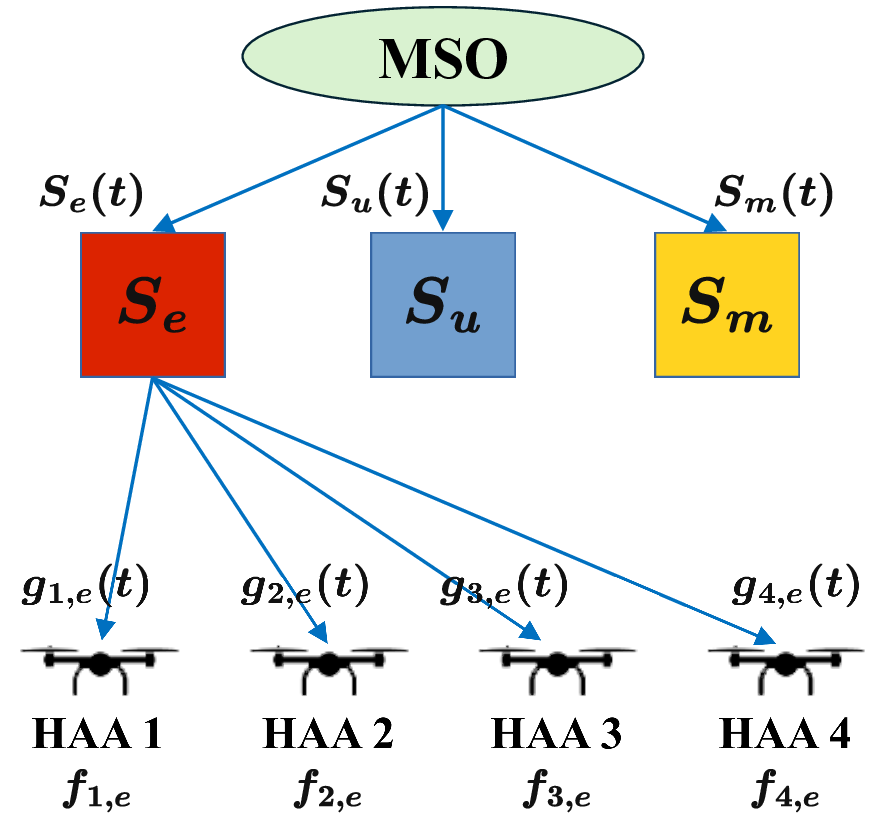} 
	\caption{The proposed dual-layer resource slicing architecture orchestrated by the macro slice orchestrator (MSO). The core-layer allocation partitions the total wireless backhaul bandwidth into service-specific macro-slices ($S_e(t)$, $S_u(t)$, and $S_m(t)$). Subsequently, the HAA-layer sub-slice allocation dynamically distributes these isolated resources across spatially deployed HAAs (e.g., $g_{n,e}(t)$ for the eMBB slice) to serve localized traffic demands.}
	\label{fig:dual_slicing_layer}
\end{figure}

At the HAA layer, each HAA $n$ further divides its allocated macro-layer slice $S_i(t)$ among its locally served users of service type $i$. Let $g_{n,i}(t)$ denote the HAA-layer sub-slice ratio for service type $i$ at HAA $n$ during control period $t$. For each service type $i \in \{e,u,m\}$, the vector $\mathbf{g}_i(t) = [g_{1,i}(t), g_{2,i}(t), \dots, g_{N,i}(t)]$ must satisfy:
\begin{align}
	& g_{n,i}(t) \geq 0, \quad \forall n = 1, 2, \dots, N, \label{eq:uav_subslice_nonnegative}\\
	& \sum_{n=1}^{N} g_{n,i}(t) = 1, \quad \forall i \in \{e,u,m\}. \label{eq:uav_subslice_sum}
\end{align}
Here, constraint~\eqref{eq:uav_subslice_nonnegative} ensures that each HAA's sub-slice ratio is non-negative, while constraint~\eqref{eq:uav_subslice_sum} guarantees that the total sub-slice ratios across all HAAs for each service type sum to unity, reflecting the complete distribution of macro-layer slices at the aerial layer.

With~\eqref{eq:core_slice_sum} and~\eqref{eq:uav_subslice_sum}, the effective resource proportion allocated to service type $i$ at HAA $n$ during control period $t$ can be expressed as:
\begin{align}
	f_{n,i}(t) = S_i(t) \cdot g_{n,i}(t), \quad \forall i \in \{e,u,m\}, \label{eq:effective_resource_proportion}
\end{align}
where $f_{n,i}(t)\in [0,1]$. Suppose the total available bandwidth is $B_{\rm total}$. Then, the bandwidth allocated to service type $i$ at HAA $n$ during control period $t$ is given by:
\begin{align}
	B_{n,i}(t) = f_{n,i}(t) \cdot B_{\rm total}, \quad \forall i \in \{e,u,m\}.  \label{eq:bandwidth_alloc}
\end{align}

\subsection{QoS Satisfaction Model}
Let $N$ denote the number of HAAs and $t$ denote the index of the control time slot. For HAA $n$ ($n=1,2,\dots,N$), the eMBB throughput $R_{n,e}(t)$ during time slot $t$ is calculated using the Shannon capacity theorem as:
\begin{equation}
    R_{n,e}(t) = B_{n,e}(t)\,\log_{2}\!\Bigl(1+\frac{P_{n}\,|h_{n}(t)|^{2}}{N_{0}\,B_{n,e}(t)}\Bigr), \label{eq:qos_re}
\end{equation}
where $B_{n,e}(t)$ is the bandwidth allocated to eMBB at HAA $n$, $P_{n}$ denotes the constant directional downlink transmit power allocated by the macro base station to the backhaul link of HAA $n$, $|h_{n}(t)|^{2}$ is the instantaneous channel gain, and $N_{0}$ is the noise power spectral density.

To accurately capture the transmission dynamics of short packets in mission-critical scenarios, the URLLC transmission rate $R_{n,u}(t)$ is modeled via the \textit{finite-blocklength approximation}~\cite{Durisi2016} as:
\begin{equation}
\begin{split}
    R_{n,u}(t) &= B_{n,u}(t) \Bigg[ \log_2\left(1 + \Gamma_{n,u}(t)\right) \\
    &\quad - \sqrt{\frac{V(\Gamma_{n,u}(t))}{L_{n,u}}} \frac{Q^{-1}(\epsilon_{u})}{\ln 2} \Bigg],
\end{split}
\label{eq:qos_ru}
\end{equation}
where $\Gamma_{n,u}(t) = \frac{P_n |h_n(t)|^2}{N_0 B_{n,u}(t)}$ represents the instantaneous \textit{signal-to-noise ratio} (SNR), $V(\Gamma_{n,u}(t)) = 1 - (1+\Gamma_{n,u}(t))^{-2}$ denotes the channel dispersion, $L_{n,u}$ is the URLLC packet blocklength, $Q^{-1}(\cdot)$ signifies the inverse Q-function, and $\epsilon_{u}$ bounds the decoding error probability. To establish the cross-layer mapping with the user-plane traffic volume $V_{k,u}(t)$ introduced in~\eqref{eq:delta_nii}, the parameter $L_{n,u}$ represents the standardized physical transport-block size into which this arriving data volume $V_{k,u}(t)$ is segmented for backhaul transmission. Unlike the volatile, time-varying traffic volume $V_{k,u}(t)$, the physical blocklength $L_{n,u}$ is a time-invariant protocol configuration predetermined during network initialization to eliminate dynamic packet-formatting signaling overhead and guarantee deterministic processing delays.

Furthermore, by characterizing the edge buffer of each HAA as an M/M/1 queueing system~\cite{She2018}, the total URLLC latency $D_{n,u}(t)$ is rigorously formulated to represent the complete system sojourn time, which inherently encompasses both the queueing wait time and the physical transmission delay:
\begin{equation}
    D_{n,u}(t) = \frac{1}{\mu_{n,u}(t) - \lambda_{n,u}(t)}, \label{eq:qos_du}
\end{equation}
where $\mu_{n,u}(t) = \frac{R_{n,u}(t)}{L_{n,u}}$ denotes the corresponding service rate, and $\lambda_{n,u}(t) = \frac{N_{\text{pkt},n,u}(t)}{T_C}$ represents the average packet arrival rate during the control period, with $N_{\text{pkt},n,u}(t) = \frac{\Delta_{n,u}(t)}{L_{n,u}}$ signifying the number of standardized physical packets generated at HAA $n$ for URLLC services within time slot $t$.

The mMTC success rate $G_{n,m}(t)$ at HAA $n$ during time slot $t$ is defined as the fraction of successfully received packets within the maximum tolerable delay $T_{\max}^m$ and maximum retransmissions $K_{\max}$:
\begin{equation}
    G_{n,m}(t) = \frac{1}{N_{\text{pkt},n,m}(t)} \sum_{j=1}^{N_{\text{pkt},n,m}(t)}\!\Bigl[T_{j} \le T_{\max}^m\;\land\;\mathrm{RTX}_{j} \le K_{\max}\Bigr], \label{eq:qos_gm}
\end{equation}
where $N_{\text{pkt},n,m}(t) = \frac{\Delta_{n,m}(t)}{L_{n,m}}$ is the total number of physical packets generated at HAA $n$ for mMTC services, $T_{j}$ is the delay of packet $j$, and $\mathrm{RTX}_{j}$ is the number of retransmissions for packet $j$. To contextualize the scale relationship, the packet-level latency threshold satisfies $T_{\max}^m < T_C$, meaning that the microscopic transmission and retransmission processes of all packets generated within slot $t$ are fully completed at the HAA edge plane before the MSO initiates the next macro control period. Structurally, the total delay $T_j$ is analyzed as the summation of the over-the-air transmission delay, propagation delay, and accumulated backoff waiting times across multiple attempts. Furthermore, the retransmission counter $\mathrm{RTX}_j$ is modeled via a standard \textit{automatic repeat request} (ARQ) protocol mechanism. Each decoding failure at the receiving end triggers a retransmission attempt and increments $\mathrm{RTX}_j$ by $1$. If successful decoding is not achieved within $\mathrm{RTX}_j \le K_{\max}$ or if the cumulative delay exceeds $T_j > T_{\max}^m$, the packet is permanently dropped, thereby reducing the aggregated success rate $G_{n,m}(t)$.

To quantify whether each service meets or violates its \textit{service level agreement} (SLA) threshold, we convert the continuous QoS indicators into normalized satisfaction levels in the interval $[0,1]$. For HAA $n$ at control period $t$, the satisfaction functions corresponding to the observed throughput, latency, and success rate are defined as:
\begin{align}
    \label{re_SLA}
    r_{n,e}(t) &= 
    \begin{cases}
        1, & R_{n,e}(t) \ge R_{\min,e},\\
        0, & R_{n,e}(t) < R_{\min,e},
    \end{cases} \\
    \label{ru_SLA}
    r_{n,u}(t) &= 
    \begin{cases}
        1, & D_{n,u}(t) \le D_{\max,u},\\
        0, & D_{n,u}(t) > D_{\max,u},
    \end{cases} \\
    \label{rm_SLA}
    r_{n,m}(t) &= 
    \begin{cases}
        1, & G_{n,m}(t) \ge G_{\min,m},\\
        0, & G_{n,m}(t) < G_{\min,m}.
    \end{cases}
\end{align}
Here, $R_{\min,e}$ denotes the minimum acceptable eMBB rate, $D_{\max,u}$ is the maximum tolerable URLLC latency, and $G_{\min,m}$ is the minimum acceptable packet success rate for mMTC.

Finally, the global time-average satisfaction of each service type is evaluated after resource allocation is completed over $T$ consecutive periods. The satisfaction values are aggregated across all $N$ HAAs and normalized by the total number of observations $NT$, yielding the average performance over the entire observation window expressed as follows:
\begin{equation}
\bar r_{e} = \frac{1}{NT}\sum_{t=1}^{T} \sum_{n=1}^{N} r_{n,e}(t), 
\end{equation}
\begin{equation}
\bar r_{u} = \frac{1}{NT}\sum_{t=1}^{T} \sum_{n=1}^{N} r_{n,u}(t), 
\end{equation}
\begin{equation}
\bar r_{m} = \frac{1}{NT}\sum_{t=1}^{T} \sum_{n=1}^{N} r_{n,m}(t).
\end{equation}

\subsection{Problem Formulation}
In this study, we model the resource allocation between the macro layer and the aerial layer as a two-tier optimization problem. The macro layer determines the allocation proportions of the three service types (eMBB, URLLC, and mMTC) from the overall resource pool, while the HAA layer further refines the resource proportions for each service type according to the macro-layer decision.

Specifically, the macro-layer ratio $S_i(t)$ determines the global portion assigned to each service type, and the HAA-layer ratio $g_{n,i}(t)$ further distributes the corresponding resources among HAAs from \eqref{eq:core_slice_nonnegative} to \eqref{eq:bandwidth_alloc}. This multi-level allocation ensures a reasonable distribution of resources for each service type while accounting for the practical demands and capabilities of HAAs. Based on this, we incorporate the satisfaction functions $r_{n,e}$, $r_{n,u}$, and $r_{n,m}$ defined from \eqref{eq:qos_re} to \eqref{eq:qos_gm}, and examine whether the QoS requirements are satisfied.

This work aims to jointly maximize the overall satisfaction of the three service types while ensuring that all allocation ratios satisfy the physical boundaries. Therefore, the proposed optimization framework is formulated as:
\begin{align}
    \max_{\mathbf{S}(t), \mathbf{g}_i(t)} \quad & \sum_{\forall i\in\{e, u, m\}} \bar r_{i}, \tag{P1} \label{eq:opt_objective}\\
    \text{s.t.} \quad & \eqref{eq:core_slice_nonnegative}, \eqref{eq:core_slice_sum}, \eqref{eq:uav_subslice_nonnegative}, \eqref{eq:uav_subslice_sum}.\notag
\end{align}
where $\mathbf{S}(t)$ is the macro-layer ratio vector and $\mathbf{g}_i(t)$ is the HAA-layer ratio vector. Note that the total backhaul capacity of the centralized macro base station is implicitly bounded by the total available bandwidth $B_{\text{total}}$. Because the dual-layer resource orchestration constraints strictly enforce $\sum_{n=1}^{N} \sum_{i \in \{e,u,m\}} f_{n,i}(t) = 1$, the aggregate throughput of the entire network is mathematically capped by the physical spectrum limit, preventing any backhaul resource violation or unfeasible allocation over the transport stratum.

The objective function in \eqref{eq:opt_objective} maps the heterogeneous service performance indicators from \eqref{re_SLA} to \eqref{rm_SLA} onto a unified evaluation axis to maximize SLA compliance across all HAAs. However, the non-stationary 6G traffic dynamics and strict dual-layer resource coupling render \eqref{eq:opt_objective} a non-linear fractional multi-dimensional knapsack problem, which is strictly NP-hard. Furthermore, the threshold-triggered SLA boundaries introduce severe non-smoothness and discontinuity into the objective space, invalidating standard gradient-based optimization solvers. 

Traditional numerical optimization methods, such as branch-and-bound or iterative convex relaxations, are physically prohibitive in this context. They require perfect deterministic environmental foresight and incur excessive execution latencies that violate the millisecond-level coherence time of URLLC services. Consequently, these traditional solvers are excluded from the scope of this study, as they cannot be deployed for online real-time execution in non-stationary aerial environments.

\section{The Proposed RAPO-TD3 Framework}
\label{sec:proposed_algorithm}
To solve the formulated non-convex MINLP problem \eqref{eq:opt_objective}, we propose an intelligent resource management framework based on the TD3 algorithm~\cite{FujimotoHM18}. This choice is motivated by the capability of TD3 to mitigate policy overestimation bias in continuous action spaces, which is critical for maintaining training stability in high-dimensional HAA environments.

\subsection{DRL Component Definitions}
We define the state space, action space, and reward function to capture the dynamics of the cognitive backhaul network.

\subsubsection{State Space Design}
To provide the TD3 agent with complete network state feedback, we transition from a basic environment snapshot to a comprehensive 61-dimensional state vector $\mathbf{s}_t$. The design philosophy of $\mathbf{s}_t$ is to encapsulate immediate physical-layer dynamics, strategic network memory, SLA boundaries, and temporal regularities. The global state vector is structured as:
\begin{align}
    \mathbf{s}_t = \left\{ \{\mathbf{o}_{n,t}\}_{n=1}^N, \mathbf{S}(t-1), \{\mathbf{g}_n(t-1)\}_{n=1}^N, \boldsymbol{\psi}_t \right\},
\end{align}
where $N$ denotes the number of HAAs. The state vector consists of three functional modules:
\begin{itemize} 
    \item \textit{Multi-Dimensional Environmental Observations}: For each HAA indexed by $n=1, \dots, N$, an 11-dimensional observation vector $\mathbf{o}_{n,t}$ encapsulates localized network dynamics. To ensure scale uniformity, it incorporates the dimensionless multi-service demand ratios $\delta_{n,e}(t)$, $\delta_{n,u}(t)$, and $\delta_{n,m}(t)$ comprising the vectorized demand state $\boldsymbol{\delta}_n(t)$, alongside the composite instantaneous backhaul channel gain $|h_n(t)|$. The remaining seven dimensions encode static service level agreement boundaries to contextualize volatile traffic loads against contract thresholds, preventing feature dominance and gradient saturation. To preserve the fundamental \textit{Markov property} without exponentially expanding the state dimensionality via recurrent networks, the single-step feedback loop serves as a mathematically sufficient statistic to guarantee policy continuity.
    
    \item \textit{Hierarchical Decision Feedback}: To handle interdependencies within the dual-layer architecture, the agent tracks its historical decisions. This 15-dimensional component consists of the preceding macro-layer slice ratios $\mathbf{S}(t-1)$ (3 dimensions) and the individual HAA-layer sub-slice weights $\mathbf{g}_n(t-1)$ for $n=1, \dots, N$ (12 dimensions). Rather than utilizing a longer temporal history window or recurrent networks which would exponentially expand the state dimensionality and violate the fundamental Markov property of the optimization framework, this single-step 1-indexed feedback loop serves as a mathematically sufficient statistic. It provides the agent with necessary strategic memory to ensure policy continuity and smooth incremental adjustments, effectively preventing tracking oscillations during real-time execution.
    
    \item \textit{Cyclic Temporal Context}: To model 24-hour periodic traffic variations, a 2-dimensional cyclic temporal feature $\boldsymbol{\psi}_t = [\sin(2\pi t / 24), \cos(2\pi t / 24)]$ is integrated. This trigonometric encoding enables proactive slice priority adaptation before anticipated peak traffic surges.
\end{itemize}

\subsubsection{Action Space and Two-Tier Projection}
The action vector $\mathbf{a}_t$ represents the multi-stage decision-making of the resource management framework, encompassing both macro-layer provisioning and HAA-layer spatial distribution. The continuous action space features a total dimensionality of $3 + 3N$, expressed as:
\begin{align}
    \mathbf{a}_t = \bigl[ \mathbf{z}_{\text{core}}(t), \mathbf{z}_{1}(t), \dots, \mathbf{z}_{N}(t) \bigr],
\end{align}
where $\mathbf{z}_{\text{core}}(t) \in \mathbb{R}^3$ corresponds to the raw outputs for the macro-layer slices, and $\mathbf{z}_{n}(t) = [z_{n,e}, z_{n,u}, z_{n,m}]$ denotes the distribution weights for service types at the $n$-th HAA.

To ensure that the agent actions strictly adhere to the bandwidth capacity constraints in~\eqref{eq:core_slice_sum} and~\eqref{eq:uav_subslice_sum} without sacrificing the differentiability of the network, we implement a \textit{double soft-max projection} mechanism. This nested transformation converts the raw actor network outputs into physically feasible allocation ratios through the following stages:
\begin{itemize}
    \item \textit{Macro-Layer Slice Projection}: The first three elements $\mathbf{z}_{\text{core}}(t)$ are processed via a soft-max function to obtain the global slice ratios $S_{i}(t)$ for eMBB, URLLC, and mMTC services. To prevent resource starvation and preserve dense gradient propagation, we apply a linear scaling:
    \begin{equation}
        S_{i}(t) = (1 - 3\rho) \cdot \text{soft-max}(z_{i}(t)) + \rho,
    \end{equation}
    where $z_{i}(t)$ denotes the continuous action value within $\mathbf{z}_{\text{core}}(t)$, and $\rho$ represents the minimum reserved slice ratio. This stage ensures structural compliance with constraint~\eqref{eq:core_slice_sum}. 
    \item \textit{HAA-Layer Distribution Projection}: For each HAA $n$, the agent determines the relative spatial weights across the three services. By applying a local soft-max operation over $\mathbf{z}_{n}(t)$, the framework standardizes the spatial priorities and ensures that the sub-slice ratios $g_{n,i}(t)$ for each service type sum to unity across all HAAs, thus satisfying constraint~\eqref{eq:uav_subslice_sum}.
    \item \textit{Deterministic Resource Mapping}: The final bandwidth $B_{n,i}(t)$ allocated to slice $i$ at UAV $n$ is determined by~\eqref{eq:bandwidth_alloc} using the effective resource proportion $f_{n,i}(t)$ from~\eqref{eq:effective_resource_proportion}.
\end{itemize}

To clarify the technical necessity of the Double Soft-max Projection layer, we explicitly analyze the dual-layer coupling constraints that cause standard independent projections to fail. Let $\mathbf{u}(t)$ denote the macro-tier allocation vector across HAAs and $\mathbf{v}_n(t)$ represent the micro-tier service slice distribution vector localized within HAA $n$. The joint resource orchestration action must satisfy a rigid hierarchical conservation law, where the absolute physical capacity allocated to a specific network slice is bounded by the multiplicative product of these two vectors. If standard independent softmax layers were utilized, the policy network would output separate vectors lacking physical scaling coherence. 
Enforcing the joint boundary via post-hoc clipping or rule-based truncation introduces non-differentiable operations that break the continuous policy gradient flow, leading to severe training instability. By contrast, the Double Soft-max Projection layer implements a nested conditionally dependent tensor mapping that explicitly enforces the joint dual-layer simplex constraint by construction. This specialized mathematical architecture preserves continuous differentiability across both resource strata, allowing exact policy gradients to backpropagate seamlessly through the dual-layer parameters to accelerate convergence.

\subsection{Reward Function Design and RAPO Mechanism}
\label{sec:reward_function}
To guide the agent toward an effective policy that balances heterogeneous service satisfaction with physical feasibility, we design a reward function characterized by differentiable utility shaping and the RAPO mechanism. This approach addresses the sparse-gradient issues inherent in binary QoS indicators while ensuring strict adherence to mission-critical constraints under volatile workloads.

\subsubsection{Differentiable Satisfaction Functions}
To facilitate effective gradient-based learning, we replace conventional binary QoS indicators with differentiable, piecewise-linear satisfaction functions $\hat{r}_{n,i}(\cdot) \in [0, 1]$. This design allows the agent to receive continuous feedback even when the SLA requirements are not fully met.
\begin{itemize}
    \item \textit{eMBB Satisfaction}: For eMBB services, the satisfaction level is driven by the achieved data rate $R_{n,e}(t)$. The utility is mapped between the minimum survival rate $R_{\min,e}$ and the target request rate $R_{{\rm req},e}$ as follows:
    \begin{equation}\label{eq:reward_e}
        \hat r_{n,e}(t) = \begin{cases}
            0, & R_{n,e}(t) \le R_{\min,e},\\
            \frac{R_{n,e}(t)-R_{\min,e}}{R_{{\rm req},e}-R_{\min,e}}, & R_{\min,e} < R_{n,e}(t) \le R_{{\rm req},e}, \\
            1, & R_{n,e}(t) > R_{{\rm req},e}.
        \end{cases}
    \end{equation}
    
    \item \textit{URLLC Satisfaction}: Given the latency-sensitive nature of URLLC, the satisfaction function is defined by the observed delay $D_{n,u}(t)$. The utility remains at its peak for delays below the target $D_{{\rm req},u}$ and drops to zero once the delay exceeds the maximum tolerable threshold $D_{\max,u}$:
    \begin{equation}\label{eq:reward_u}
        \hat r_{n,u}(t) =\begin{cases}
            0, & D_{n,u}(t) > D_{\max,u}, \\
            \frac{D_{\max,u} - D_{n,u}(t)}{D_{\max,u} - D_{{\rm req},u}}, & D_{{\rm req},u} < D_{n,u}(t) \le D_{\max,u}, \\
            1, & D_{n,u}(t) \leq D_{{\rm req},u}. 
        \end{cases}
    \end{equation}
    
    \item \textit{mMTC Satisfaction}: For mMTC services, the satisfaction level is determined by the packet success rate $G_{n,m}(t)$. The utility reflects the reliability of massive device connectivity between the minimum acceptable level $G_{\min,m}$ and the target reliability $G_{{\rm req},m}$:
    \begin{equation}\label{eq:reward_m}
        \hat r_{n,m}(t) = \begin{cases}
            1, & G_{n,m}(t) \geq G_{{\rm req},m}, \\
            \frac{G_{n,m}(t) - G_{\min,m}}{G_{{\rm req},m} - G_{\min,m}}, & G_{\min,m} \le G_{n,m}(t) < G_{{\rm req},m}, \\
            0, & G_{n,m}(t) < G_{\min,m}.
        \end{cases}
    \end{equation}
\end{itemize}

\subsubsection{Global Reward and Differentiable Penalties}
While conventional approaches rely on heuristic penalty terms to discourage boundary violations, they often introduce non-differentiable artifacts. In contrast, our double soft-max projection layer enforces simplex constraints by construction. Consequently, the formulation inherently obviates the need for explicit bandwidth boundary penalties, thereby preserving the gradient continuity required for stable convergence.

Therefore, the system penalty $P_{\text{system}}(t)$ is purely dedicated to penalizing severe URLLC latency violations, ensuring that mission-critical links are safeguarded:
\begin{equation}
    P_{\text{system}}(t) = \varpi \sum_{n=1}^{N} \left( \frac{D_{n,u}(t) - D_{\max,u}}{D_{\max,u}} \right)^{+}
\end{equation}
where $\varpi$ signifies the URLLC delay penalty coefficient regulating the scaling severity of delay violations, and $(\cdot)^{+}$ denotes the positive part function. Without loss of generality, $\varpi$ is parameterized as $1.0$ to serve as the normalized reference baseline for the penalty domain, thereby enabling the reward weight vectors $\boldsymbol{\omega}_{\text{base}}$ and $\boldsymbol{\omega}_{\text{alert}}$ to be calibrated systematically relative to this unified penalty scaling. The global reward $r(t)$ observed by the agent at time slot $t$ is then formulated as the weighted aggregate of slice satisfactions minus the time-varying penalty:
\begin{equation} \label{eq:global_reward}
    r(t) = \sum_{n=1}^{N} \sum_{i \in \{e,u,m\}} \omega_i(t) \cdot \hat r_{n,i}(t) - \chi(t) \cdot P_{\text{system}}(t),
\end{equation}
where $\omega_i(t)$ represents the dynamic service weights governed by the RAPO mechanism, and $\chi(t)$ is the dynamic penalty scaling factor. To prioritize exploration in early training while ensuring strict constraint adherence during convergence, we implement a cyclical penalty factor $\chi(t)$ defined as:
\begin{align}
    \chi(t) & = \left(\chi_{\text{init}} + \chi_{\text{grow}} \times \frac{EP_{\rm current}}{EP_{\rm total}}\right) \notag\\
    & \times \left(1 + \chi_{\text{osc}} \cdot \sin\left(2\pi \cdot \frac{EP_{\rm current}}{T_{\text{osc}}}\right)\right),
\end{align}
where $EP_{\rm current}$ and $EP_{\rm total}$ are the current and total episode indices, respectively. The parameters $\{\chi_{\text{init}}, \chi_{\text{grow}}, \chi_{\text{osc}}, T_{\text{osc}}\}$ control the baseline, growth rate, oscillation amplitude, and oscillation period.

\subsubsection{Resilient Adaptive Priority Orchestration Mechanism}
A distinctive architectural feature of our framework is the RAPO mechanism governed by an adaptive smooth gating structure. This enables the agent to shift its priorities fluidly under extreme network pressure without inducing training oscillations. Based on the real-time penalty magnitude, the system autonomously transitions between two conceptual operational modes:
\begin{itemize}
    \item \textit{Balanced Mode}: When the network is stable ($P_{\text{system}}(t) \le P_{\text{thresh}}$), the system adopts the base weight vector $\boldsymbol{\omega}_{\text{base}}$ to maximize eMBB utility while maintaining baseline URLLC stability.
    \item \textit{Alert Mode}: Upon detecting significant QoS violations ($P_{\text{system}}(t) > P_{\text{thresh}}$), the system prioritizes the alert weight vector $\boldsymbol{\omega}_{\text{alert}}$ to strictly safeguard ultra-reliable transmissions.
\end{itemize}

Instead of an abrupt, non-differentiable mode switch, the dynamic weight vector $\boldsymbol{\omega}(t) = [\omega_e(t), \omega_u(t), \omega_m(t)]^\top$ is derived by seamlessly blending the base and alert mode vectors:
\begin{equation}
    \boldsymbol{\omega}(t) = (1 - \xi(t)) \cdot \boldsymbol{\omega}_{\text{base}} + \xi(t) \cdot \boldsymbol{\omega}_{\text{alert}},
\end{equation}
where $\xi(t) \in (0, 1)$ serves as the adaptive interpolation factor. To ensure a smooth transition, $\xi(t)$ is modeled via a sigmoid function evaluated against the instantaneous system penalty:
\begin{equation}\label{eq:sigmoid_gate}
    \xi(t) = \frac{1}{1 + \exp\left(-\frac{P_{\text{system}}(t) - P_{\text{thresh}}}{\kappa}\right)}.
\end{equation}
Here, $P_{\text{thresh}}$ represents the critical penalty threshold that triggers the alert state, and $\kappa$ acts as the temperature parameter controlling the steepness of the mode transition. This mathematical design enables an elastic priority adaptation where best-effort services are dynamically and gracefully scaled down to safeguard the mission-critical URLLC slices during extreme traffic demand surges.

\subsection{RAPO-TD3 Implementation Details}
To ensure reliable resource management in the non-stationary aerial backhaul environment, we detail the implementation of the proposed RAPO-TD3 framework. Our approach mitigates the overestimation bias inherent in standard actor-critic methods by maintaining twin critic networks, $Q_{\phi_1}$ and $Q_{\phi_2}$, and adopting the minimum of their estimates for value updates. 

Furthermore, to strike a critical balance between global architectural stability and localized edge node exploration, we introduce a \textit{selective exploration noise} mechanism. Specifically, the Ornstein-Uhlenbeck (OU) exploration noise is exclusively injected into the HAA-layer spatial actions $\mathbf{z}_{n}(t)$, while the macro-layer ratios $\mathbf{z}_{\text{core}}(t)$ are strictly generated by deterministic policy outputs. This ensures that the MSO does not destabilize the entire network core while local HAAs explore optimal spatial distributions.

Finally, we replace uniform sampling with $\alpha$-prioritized experience replay~\cite{PER_2016}. This mechanism prioritizes transitions with larger temporal difference (TD) errors, which correspond to rare critical states where QoS requirements are violated. During execution, the RAPO mechanism dynamically calibrates the multi-service priority weights to compute the global reward, seamlessly guiding the agent to safeguard URLLC constraints. The integration of the double soft-max projection, selective noise, and the RAPO mechanism is designed so that every evaluated action satisfies the mathematical simplex bounds while accelerating convergence. The complete algorithmic procedure is detailed in Algorithm~\ref{alg:rapo_td3}.

\begin{algorithm2e}[!t]
\small
\SetAlgoLined
\SetKwInOut{Input}{Input}
\SetKwInOut{Output}{Output}

\caption{RAPO-TD3 for Dual-Layer Slicing in HAA Backhaul Networks}
\label{alg:rapo_td3}

\Input{Initial state $\mathbf{s}_1$, target update rate $\tau$, delayed update interval $d$, PER parameters $\alpha, \beta, \epsilon_{\text{per}}$, hyperparameters $\sigma_{\rm tgt}, c_{\rm clip}$, learning rates $\eta_a, \eta_c$, discount factor $\gamma$, maximum episodes $EP_{\rm total}$, total time steps $T$, mini-batch size $M$}
\Output{Optimized actor network $\pi_\theta$ for resource slicing}

Initialize critic networks $Q_{\phi_1}, Q_{\phi_2}$ and actor $\pi_\theta$ with random weights\;
Initialize target networks $\phi_1' \leftarrow \phi_1, \phi_2' \leftarrow \phi_2, \theta' \leftarrow \theta$\;
Initialize PER buffer $\mathcal{D}$\;

\For{$e = 1$ \KwTo $EP_{\rm total}$}{
    Observe initial state $\mathbf{s}_1 = \{ \{\mathbf{o}_{n,1}\}_{n=1}^N, \mathbf{S}(0), \{\mathbf{g}_n(0)\}_{n=1}^N, \boldsymbol{\psi}_1 \}$\;

    \For{$t = 1$ \KwTo $T$}{
        Select base action $\mathbf{a}_t = \pi_\theta(\mathbf{s}_t)$\;
        \tcp{Selective Exploration Noise}
        Inject OU noise $\boldsymbol{\epsilon}$ exclusively into HAA-layer sub-slice actions $\mathbf{z}_{n}(t)$\;
        
        \tcp{Double Soft-max Projection}
        Stage-1: Map $\mathbf{z}_{\text{core}}(t)$ to macro-layer ratios $S_i(t)$ such that $\sum_{i} S_i(t) = 1$\;
        Stage-2: Map $\mathbf{z}_{n}(t)$ to spatial distribution weights $g_{n,i}(t)$\;
        Calculate final bandwidth fraction: $f_{n,i}(t) = S_i(t) \cdot g_{n,i}(t)$\;
        Calculate physical bandwidth: $B_{n,i}(t) = f_{n,i}(t) \cdot B_{\rm total}$\;        
        
        \tcp{RAPO Evaluation \& Execution}
        Execute $B_{n,i}(t)$, compute dynamic priority weights via the RAPO mechanism, and observe global reward $r(t)$ by \eqref{eq:global_reward} and next state $\mathbf{s}_{t+1}$\;
        
        \tcp{PER Storage \& Priority Update}
        Store transition $(\mathbf{s}_t, \mathbf{a}_t, r(t), \mathbf{s}_{t+1})$ in $\mathcal{D}$ with maximum priority\;
        
        \tcp{Network Update Stage with Twin-Critic}
        Sample mini-batch of $M$ transitions from $\mathcal{D}$ with probability $P(b) = p_b^\alpha / \sum_j p_j^\alpha$\;
        Generate smoothed target noise $\tilde{\boldsymbol{\epsilon}} \sim \text{clip}(\mathcal{N}(0, \sigma_{\rm tgt}), -c_{\rm clip}, c_{\rm clip})$\;
        Compute target action: $\tilde{\mathbf{a}} \leftarrow \text{clip}(\pi_{\theta'}(\mathbf{s}_{t+1}) + \tilde{\boldsymbol{\epsilon}}, \mathbf{a}_{\rm low}, \mathbf{a}_{\rm high})$\;
        $y \leftarrow r_b + \gamma \min_{j=1,2} Q_{\phi_j'}(\mathbf{s}_{t+1}, \tilde{\mathbf{a}})$\;
        Update critics with learning rate $\eta_c$ by minimizing weighted MSE loss $\mathcal{L} = \frac{1}{M} \sum_{b=1}^{M} w_b (y - Q_{\phi_j}(\mathbf{s}_b, \mathbf{a}_b))^2$\;
        
        \If{$t \pmod d == 0$}{
            Update actor $\theta$ with learning rate $\eta_a$ using sampled policy gradient $\nabla_\theta J$\;
            Soft update target networks: $\theta' \leftarrow \tau\theta + (1-\tau)\theta'$\;
        }
        Update priority $p_b \leftarrow |y - Q_{\phi_1}(\mathbf{s}_b, \mathbf{a}_b)| + \epsilon_{\text{per}}$ for sampled transitions\;
    }
}
\end{algorithm2e}

\subsection{Complexity Analysis}
The computational and structural complexity of the proposed hierarchical slicing framework is analyzed across two distinct phases: offline centralized training and online decentralized execution.

During the online execution phase, the centralized MSO strictly relies on the trained actor network $\pi_\theta$ to generate real-time allocation decisions. Let $L_a$ denote the number of fully connected layers in the actor network, and $U_l$ represent the number of neurons in the $l$-th layer, where $U_0$ and $U_{L_a}$ correspond to the dimensions of the state vector $\mathbf{s}_t$ and the action vector $\mathbf{a}_t$, respectively. The asymptotic time complexity for generating a single operational decision is bounded by $\mathcal{O}(\sum_{l=0}^{L_a-1} U_l U_{l+1})$. This forward-propagation calculation requires merely a fraction of a millisecond, mathematically satisfying the strict execution constraints mandated by URLLC services.

Conversely, the offline training phase incurs a higher computational footprint. During each gradient update step, the framework computes the forward and backward passes for both the actor and the twin critic networks using a mini-batch of size $M$. Furthermore, the integration of prioritized experience replay utilizing a sum-tree data structure introduces a sampling and priority-updating overhead of $\mathcal{O}(M \log |\mathcal{D}|)$, where $|\mathcal{D}|$ is the maximum capacity of the replay buffer. Since this training complexity is exclusively processed offline by high-performance computing clusters at the core network, it remains completely decoupled from real-time aerial operations.

Regarding the signaling overhead, the MSO requires state gathering from $N$ HAAs and subsequent action broadcasting at the beginning of each control period. The dimension of the gathered state vector scales linearly with the number of agents, yielding a communication complexity of $\mathcal{O}(N)$. This strict linear scaling ensures that the signaling overhead does not exponentially explode as the aerial fleet expands, efficiently addressing the scalability requirements for dense 6G network deployments.

\subsection{Theoretical Convergence and Stability Analysis}
\label{sec:convergence_analysis}

Although providing a global convergence proof for deep reinforcement learning with non-linear neural network approximators remains an open mathematical challenge, we establish the theoretical stability and convergence properties of the proposed RAPO-TD3 framework. The convergence guarantees of our architecture are primarily driven by the properties introduced by the double soft-max projection and the RAPO mechanism.

\begin{property}[Action and Reward Boundedness]
    Let $\mathcal{A}$ denote the transformed action space and $\mathcal{R}$ denote the reward space.
    The double soft-max projection restricts the executed bandwidth allocations to the joint probability simplex, i.e., $\sum_{n=1}^{N} \sum_{i \in \{e,u,m\}} f_{n,i}(t) = 1$, with $f_{n,i}(t) \in [0,1]$.
    Consequently, the action space $\mathcal{A}$ is compact and bounded. Furthermore, as defined in \eqref{eq:reward_e} through \eqref{eq:reward_m}, the piecewise-linear satisfaction functions yield utilities bounded within $[0,1]$, and the weight vector $\boldsymbol{\omega}(t)$ is convexly interpolated.
    Therefore, the global reward $r(t)$ is bounded by $r_{\min} \le r(t) \le r_{\max}$, preventing diverging returns and bounding the temporal difference (TD) error variance during experience replay.
\end{property}

\begin{property}[Lipschitz Continuity of the Objective Landscape]
    A fundamental requirement for stable gradient descent in deterministic policy gradient frameworks is the Lipschitz continuity of the objective function~\cite{Silver2014}.
    Conventional approaches with binary QoS penalties or heuristic hard-switching rules introduce step-function discontinuities, causing infinite gradient variances.
    In our framework, both the double soft-max projection and the RAPO sigmoid gating function \eqref{eq:sigmoid_gate} are continuously differentiable ($\mathcal{C}^1$ functions) across the entire state-action space.
    Let $J(\theta)$ be the expected return objective of the actor network.
    Following standard optimization theory~\cite{Boyd2004}, the gradients of the transformations with respect to the pre-activation outputs $\mathbf{z}(t)$ are bounded by the properties of the soft-max derivative (Jacobian matrix norms are bounded by $1$) and the sigmoid derivative (bounded by $\frac{1}{4\kappa}$).
    Consequently, the overall reward function is $L_r$-Lipschitz continuous, satisfying:    
    \begin{equation}
        |r(\mathbf{s}, \mathbf{a}_1) - r(\mathbf{s}, \mathbf{a}_2)| \le L_r ||\mathbf{a}_1 - \mathbf{a}_2||_2, \quad \forall \mathbf{a}_1, \mathbf{a}_2 \in \mathcal{A},
    \end{equation}
    where $L_r > 0$ is the Lipschitz constant of the reward function.
\end{property}

Building upon the compact state-action boundary and the smooth, gradient-bounded objective landscape, we formulate the stability guarantee of the proposed learning architecture. These fundamental properties collectively eliminate the risks of gradient explosion and unbounded Q-value divergence, paving the way for the algorithmic stability theorem.

\begin{thm}[Algorithmic Stability Guarantee]
    For the proposed RAPO-TD3 framework, the variance of the sampled policy gradient $\nabla_\theta J$ is bounded. When coupled with the twin-critic overestimation clipping and a prioritized sampling probability $P(b) > 0$ for all transitions in $\mathcal{D}$, the Robbins-Monro conditions for stochastic approximation~\cite{Robbins1951} are satisfied. This implies asymptotic convergence to a stable local optimum policy under standard decaying learning rate assumptions without divergent oscillations.
\end{thm}

\begin{proof}
    To establish asymptotic convergence, we verify that the proposed architecture satisfies the Robbins-Monro conditions for stochastic approximation~\cite{Robbins1951}. First, let $Q^{\pi}(\mathbf{s}, \mathbf{a})$ denote the true action-value function under policy $\pi_\theta$. Based on the boundedness established in Property 1, the instantaneous reward is bounded within $[r_{\min}, r_{\max}]$. Given the infinite-horizon discounted return with a discount factor $\gamma \in (0,1)$, the target value $y$ and the critic estimates $Q_{\phi_j}(\mathbf{s}, \mathbf{a})$ are bounded within the compact interval $[r_{\min}/(1-\gamma), r_{\max}/(1-\gamma)]$, eliminating the risk of numerical divergence.

    Second, the deterministic policy gradient is formulated as $\nabla_\theta J = \mathbb{E}[ \nabla_\theta \pi_\theta(\mathbf{s}) \nabla_{\mathbf{a}} Q_{\phi_j}(\mathbf{s}, \mathbf{a}) |_{\mathbf{a}=\pi_\theta(\mathbf{s})} ]$. Under Property 2, the reward function is $L_r$-Lipschitz continuous, which, under the standard assumption of smooth environmental transition dynamics, inherently bounds the spatial derivative of the parameterized critic network such that $\| \nabla_{\mathbf{a}} Q_{\phi_j}(\mathbf{s}, \mathbf{a}) \|_2 \le L_{\text{critic}} < \infty$, where $L_{\text{critic}}$ represents the uniform Lipschitz constant of the critic network gradient. Simultaneously, since the actor network $\pi_\theta(\mathbf{s})$ is constructed using fully connected layers with smooth activation functions over a compact parameter space, its structural gradient is also bounded by $\| \nabla_\theta \pi_\theta(\mathbf{s}) \|_2 \le L_{\text{actor}} < \infty$, where $L_{\text{actor}}$ is the corresponding bounding constant. 
    
    During training, transitions are drawn via prioritized sampling rather than uniform distribution. By applying \textit{Importance Sampling} (IS) weights to correct the non-uniform sampling bias, the stochastic policy gradient estimate remains mathematically unbiased. Because the norm of the sample gradient mapping is bounded by the product of the Lipschitz constants, the variance of this unbiased estimator computed over a mini-batch of size $M$ scales inversely with $M$ and is structurally upper-bounded by a finite constant variance $\sigma^2$:
    \begin{equation}
        \mathbb{E}\left[ \left\| \hat{\nabla}_\theta J - \nabla_\theta J \right\|_2^2 \right] \le \frac{(L_{\text{critic}} L_{\text{actor}})^2}{M} \le \sigma^2 < \infty.
    \end{equation}

    Let $\eta_{a,t}$ denote the decaying learning rate at training step $t$. By enforcing the standard step-size criteria $\sum_{t=1}^{\infty} \eta_{a,t} = \infty$ and $\sum_{t=1}^{\infty} \eta_{a,t}^2 < \infty$, the bounded gradient variance ensures that the stochastic noise cancels out asymptotically. Furthermore, the twin-critic clipping mechanism in TD3 minimizes the positive overestimation bias, ensuring that the policy updates do not overshoot into unstable regions of the action simplex. Since the prioritized experience replay mechanism guarantees a non-zero selection probability $P(b) > 0$ for all historical transitions, the tracking sequence cannot entrap in empty gradient zones. Consequently, according to the stochastic approximation theorem, the actor network parameter vector $\theta$ converges to a stationary local optimum satisfying $\nabla_\theta J = 0$ with probability $1$.
\end{proof}

\section{Simulation Results and Analysis}
\label{sec:simulation_results}

\subsection{Simulation Setup and Traffic Profile}
\label{sec:simulation_setup}

The simulation environment is designed to evaluate the performance of the proposed RAPO-TD3 framework in a dynamic aerial backhaul scenario. To accurately capture the complex inter-cell interference and multi-agent resource contention without succumbing to the exponential state-space explosion typical of deep reinforcement learning, we model a dense interference sub-cluster comprising $N=4$ HAAs. This configuration is consistent with standard 3GPP deployment scenarios for aerial vehicles~\cite{3gpp_tr_36_777} and macro-cell sectoring interference models~\cite{3gpp_tr_38_901}. It provides a mathematically sound baseline for validating decentralized execution, as the previously proven $\mathcal{O}(N)$ computational complexity guarantees linear scalability to larger fleet deployments. 
The detailed communication parameters, QoS thresholds, and reinforcement learning hyperparameters are consolidated in Table~\ref{tab:simulation_parameters}.

\begin{table}[!t]
\caption{Consolidated Simulation and Hyperparameters}
\label{tab:simulation_parameters}
\centering
\setlength{\tabcolsep}{2pt} 
\begin{tabular}{llc}
\toprule
\textbf{Parameter Description} & \textbf{Symbol} & \textbf{Default Value} \\ 
\midrule
\multicolumn{3}{l}{\textit{Network Environment}} \\ 
\midrule
Number of HAAs & $N$ & 4 \\
Total System Bandwidth & $B_{\rm total}$ & 10 MHz \\
Centralized Macro Control Period & $T_C$ & 100 ms \\
MBS Downlink Power per HAA & $P_n$ & 23 dBm \\
Noise Power Spectral Density & $N_0$ & $1 \times 10^{-9}$ W/Hz \\
Carrier Frequency & $f_c$ & 2.5 GHz \\
HAA Operational Altitude & $H_{\text{haa}}$ & 120 m \\ 
Dense Urban Environment Constants & $a_{\text{env}}, b_{\text{env}}$ & 12.08, 0.11 \\
FSPL Bulk Normalization Constant & $C_{\text{FSPL}}$ & -147.55 \\
Additional Path Loss (LoS, NLoS) & $\eta_{\text{LoS}}, \eta_{\text{NLoS}}$ & 1.0 dB, 20.0 dB \\
Small-Scale Fading Profile & $\Omega_n(t)$ & Rayleigh ($s=1.0$) \\ 
\midrule
\multicolumn{3}{l}{\textit{Service Requirements}} \\ 
\midrule
eMBB Rate Threshold (Min, Req) & $R_{{\min},e}, R_{{\rm req},e}$ & 10 Mbps, 50 Mbps \\
URLLC Delay Threshold (Req, Max) & $D_{{\rm req},u}, D_{\max,u}$ & 2 ms, 4 ms \\
mMTC Success Rate (Min, Req) & $G_{{\min},m}, G_{{\rm req},m}$ & 0.90, 0.99 \\
URLLC Packet Size \& Error Bound & $L_{n,u}, \epsilon_u$ & 3000 bits, $10^{-5}$ \\
mMTC Packet Size & $L_{n,m}$ & 256 bits \\ 
mMTC Time Constraint & $T_{\max}^{m}$ & 20 ms \\ 
mMTC Maximum Retransmissions & $K_{\max}$ & 1 \\
URLLC Burst Probability & $P_{\text{burst}}$ & 0.20 per hour \\
URLLC Burst Duration & $T_{\text{burst}}$ & 0.2 hours \\ 
\midrule
\multicolumn{3}{l}{\textit{RAPO-TD3 Hyperparameters}} \\ 
\midrule
Actor/Critic Learning Rate & $\eta_a, \eta_c$ & $5 \times 10^{-5}, 5 \times 10^{-4}$ \\
Discount Factor & $\gamma$ & 0.99 \\
Target Update Rate & $\tau$ & 0.005 \\
Replay Buffer Capacity & $|\mathcal{D}|$ & $3\times 10^5$ \\
Mini-batch Size & $M$ & 256 \\
PER Parameters (Exponent, Const) & $\alpha, \epsilon_{\text{per}}$ & 0.6, $10^{-3}$ \\ 
Policy Smoothing (std, clip) & $\sigma_{\rm tgt}, c_{\rm clip}$ & 0.1, 0.3 \\
Delayed Actor Update Interval & $d$ & 2 \\ 
Minimum Reserved Slice Ratio & $\rho$ & 0.0 \\ 
Total Training Episodes & $EP_{\rm total}$ & 500 \\
Total Steps per Episode & $T$ & 200 \\
\midrule
\multicolumn{3}{l}{\textit{Reward \& RAPO Settings}} \\ 
\midrule
Base / Alert Mode Weights & $\boldsymbol{\omega}_{\text{base}}, \boldsymbol{\omega}_{\text{alert}}$ & $[7, 8, 2], [8, 9, 1]$ \\
Initial Penalty Threshold & $P_{\text{thresh}}$ & 0.05 \\ 
Penalty Adaptation Factor & $\kappa$ & 0.05 \\
URLLC Delay Penalty Coefficient & $\varpi$ & 1.0 \\
Penalty Factor (Init, Grow) & $\chi_{\text{init}}, \chi_{\text{grow}}$ & 0.2, 0.8 \\ 
Penalty Oscillation (Amp, Period) & $\chi_{\text{osc}}, T_{\text{osc}}$ & 0.1, 30 Episodes \\
\bottomrule
\end{tabular}
\end{table}

To mathematically evaluate the adaptive priority optimization capability governed by the proposed RAPO framework, we implement two non-stationary multi-service traffic models, as illustrated in Fig.~\ref{fig:3}. The simulation evaluates the system under the following distinct workload behaviors:
\begin{itemize}
    \item \textit{24-Hour Periodic Traffic}: As shown in Fig.~\ref{fig:3-a}, the arrival rates for eMBB, URLLC, and mMTC slices follow a non-stationary Poisson process with time-varying intensities. As depicted in the relative demand profile, the eMBB and mMTC traffic exhibit a clear diurnal cycle, peaking during daytime hours to simulate intense human-centric urban activity, while the URLLC traffic maintains a lower but strictly bounded baseline demand.
    \item \textit{Burst URLLC Traffic}: To simulate mission-critical emergencies, we introduce stochastic data bursts into the environment. A burst event occurs with an anomaly probability of $P_{\text{burst}} = 0.20$ evaluated at each operational hour. Once triggered, the localized workload surge persists for a deterministic duration of $T_{\text{burst}} = 0.2$ hours, which corresponds to 12 minutes of continuous peak pressure. During these anomalous intervals, the URLLC packet arrival rate is instantaneously scaled by a factor of $5.0$, representing a severe $500\%$ demand surge that forces the MSO to execute real-time slice priority adaptation. Fig.~\ref{fig:3-b} depicts a sample realization of this bursty traffic pattern, capturing the extreme multi-order spikes over the 24-hour temporal evaluation window.
\end{itemize}

\begin{figure}[!t]
    \centering
    \subfigure[24-Hour Periodic Traffic]{
        \label{fig:3-a}
        \includegraphics[width=.92\linewidth]{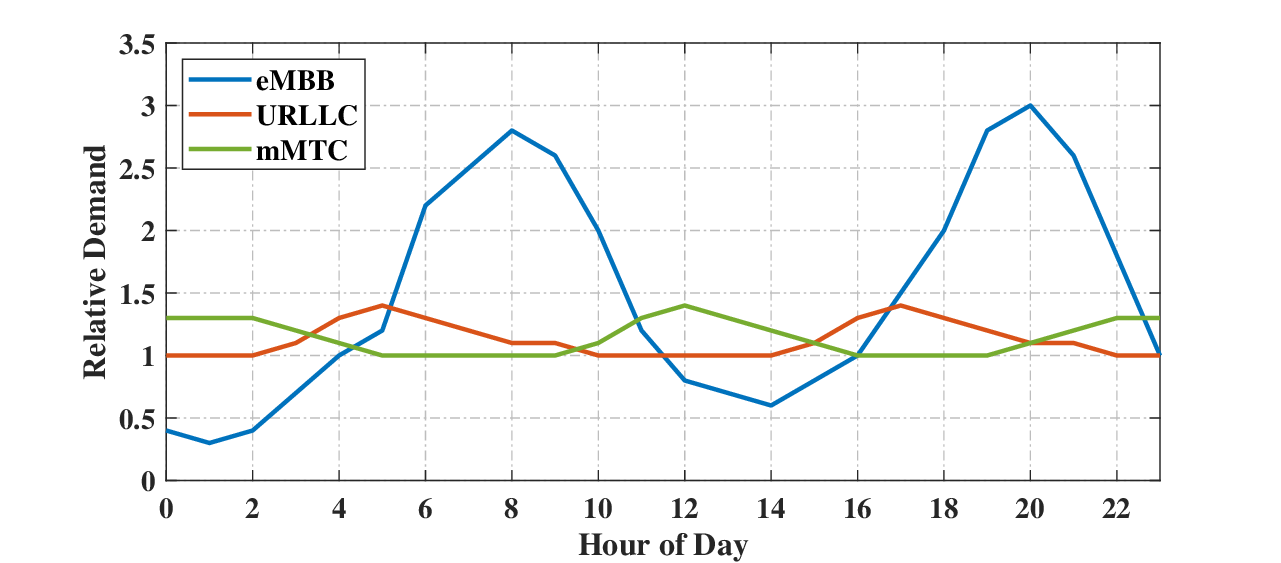}
    }\\
    \subfigure[Burst URLLC Traffic]{
        \label{fig:3-b}
        \includegraphics[width=.92\linewidth]{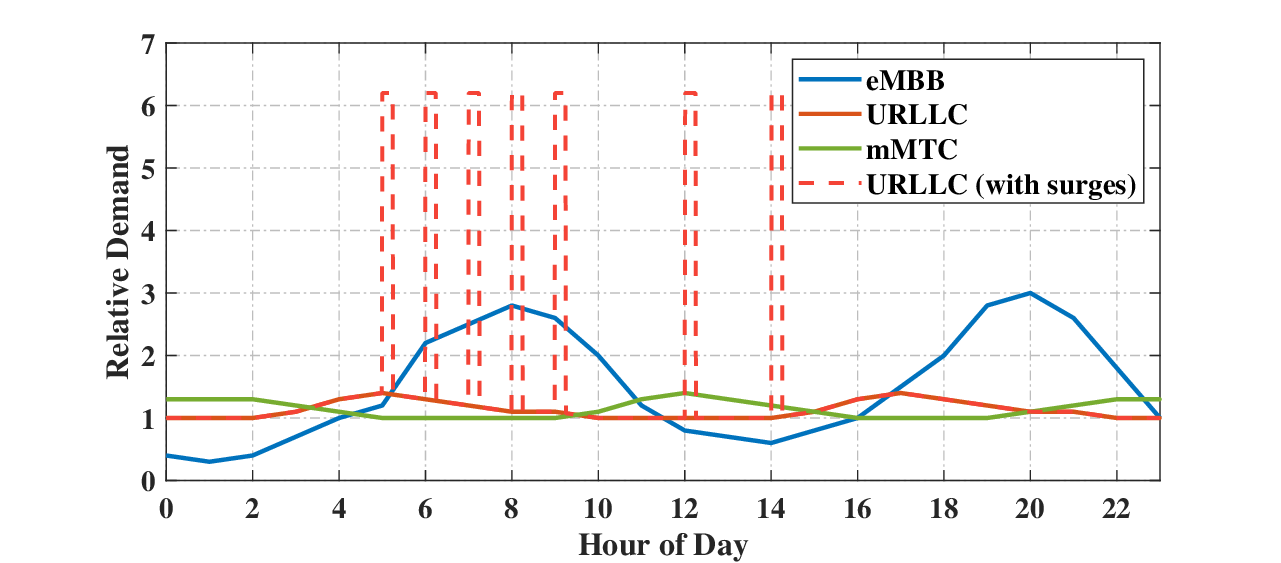}
    }
    \caption{Two different multi-service traffic demand patterns in the simulation environment: \subref{fig:3-a} 24-Hour Periodic Traffic, and \subref{fig:3-b} Burst URLLC Traffic.}
    \label{fig:3}
\end{figure}

To comprehensively demonstrate the effectiveness of the proposed RAPO-TD3 framework in handling dual-layer resource slicing and physical boundary constraints, we benchmark its performance against two state-of-the-art continuous-control reinforcement learning baselines, namely PPO and DDPG. The structural characteristics of these baseline schemes are outlined as follows:
\begin{itemize}
    \item \textit{PPO Baseline}: This scheme implements the PPO algorithm, which functions as an on-policy actor-critic framework utilizing a clipped surrogate objective function to bound policy updates. In this evaluation, PPO represents an unconstrained reward-maximization approach that lacks the specialized hierarchical projection architecture, thereby serving as a primary reference to evaluate the necessity of embedding structural constraints directly into the action space.
    \item \textit{DDPG Baseline}: This scheme implements the DDPG algorithm, which optimizes a deterministic policy through the gradients of a single critic network. This baseline operates without target policy smoothing or delayed actor updates, allowing us to isolate and verify the algorithmic efficacy of the twin-critic overestimation mitigation mechanisms inherent to the TD3 architecture under highly volatile environments.
\end{itemize}

\subsection{Convergence Performance and Constraint Satisfaction}
\label{sec:convergence_performance}

To validate the learning efficiency and the constraint adherence capabilities of the proposed framework, we first evaluate the algorithmic convergence under the normal 24-hour periodic traffic scenario. The performance of the proposed RAPO-TD3 algorithm is benchmarked against two widely adopted continuous-control DRL algorithms: DDPG and PPO. The training dynamics, including the cumulative reward, system penalty, and mission-critical satisfaction ratios over 500 episodes, are illustrated in Fig.~\ref{fig:convergence}.

An initial observation of the total cumulative reward in Fig.~\ref{fig:convergence:a} reveals that the PPO baseline numerically achieves a competitive raw reward score. However, the system penalty dynamics in Fig.~\ref{fig:convergence:b} expose a critical operational flaw. Because PPO lacks the structural double soft-max projection to accurately navigate the coupled multi-tier conservation laws, it fails to optimize the spatial distribution effectively. This structural deficiency causes PPO to over-allocate resources to eMBB services at the expense of mission-critical slices, triggering severe URLLC latency violations and resulting in penalty spikes reaching up to 0.35.

Conversely, the proposed RAPO-TD3 framework converges to a stable, constraint-aware operating point. Driven by the architectural double soft-max projection and the dynamic penalty scaling factor $\zeta(t)$, the RAPO-TD3 agent is able to navigate the complex non-convex action space without relying on constraint-violating shortcuts. As depicted in Fig.~\ref{fig:convergence:b}, the system penalty for RAPO-TD3 is aggressively suppressed to a near-zero level after the initial exploration phase, maintaining an exceptionally low and safe footprint compared to the baseline algorithms.

The effectiveness of the proposed constraint-aware design is further supported by the results shown in Fig.~\ref{fig:convergence:c}, which tracks the URLLC QoS satisfaction ratio. While the PPO baseline frequently drops below the mandatory reliability threshold, declining to approximately 0.92 during peak traffic hours, the proposed RAPO-TD3 framework consistently safeguards the latency-sensitive URLLC slices. Furthermore, Fig.~\ref{fig:convergence:d} explicitly illustrates the convergence trajectory of the mMTC QoS satisfaction ratio under this baseline scenario to address the holistic multi-service evaluation. In complete alignment with the reward and penalty dynamics, the proposed RAPO-TD3 framework tightly tracks the near-optimal 1.0 satisfaction boundary throughout the entire training process with minor statistical fluctuations. Conversely, the PPO baseline displays pronounced performance degradation, dropping down to a satisfaction level of 0.985 near training episodes 225 and 450. This unstable behavior further substantiates that unconstrained reward-maximization schemes compromise best-effort slice stability even under non-bursty diurnal workloads, whereas our constraint-aware architecture guarantees stable multi-service isolation and steady-state protocol compliance.

\begin{figure}[!t]
    \centering
    \subfigure[Convergence of Total Reward]{
        \label{fig:convergence:a}
        \includegraphics[width=0.48\linewidth]{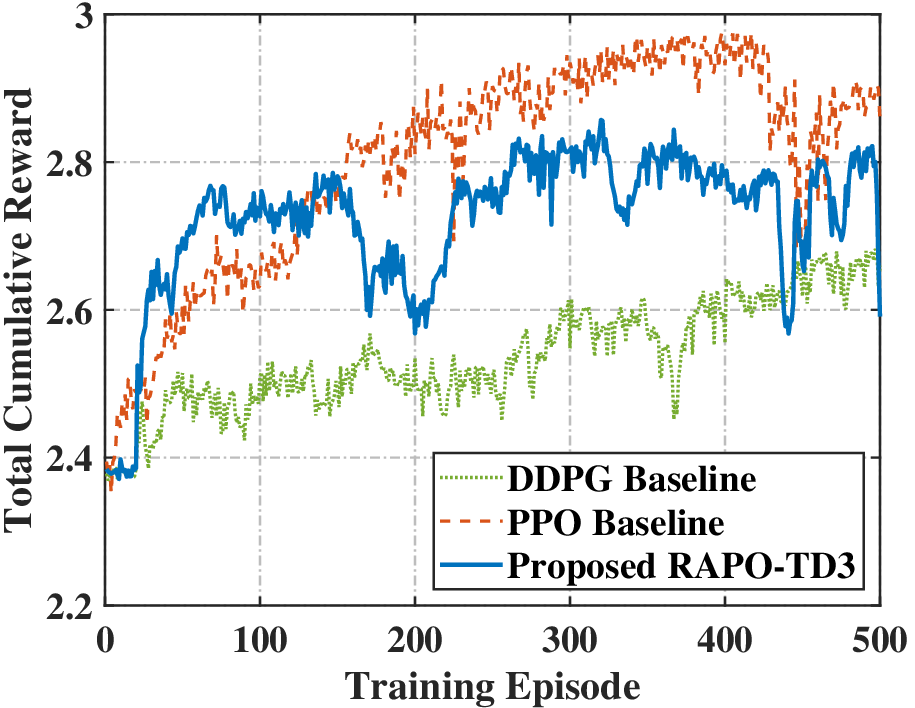}
    }%
    \subfigure[Constraint Violation Penalty]{
        \label{fig:convergence:b}
        \includegraphics[width=0.48\linewidth]{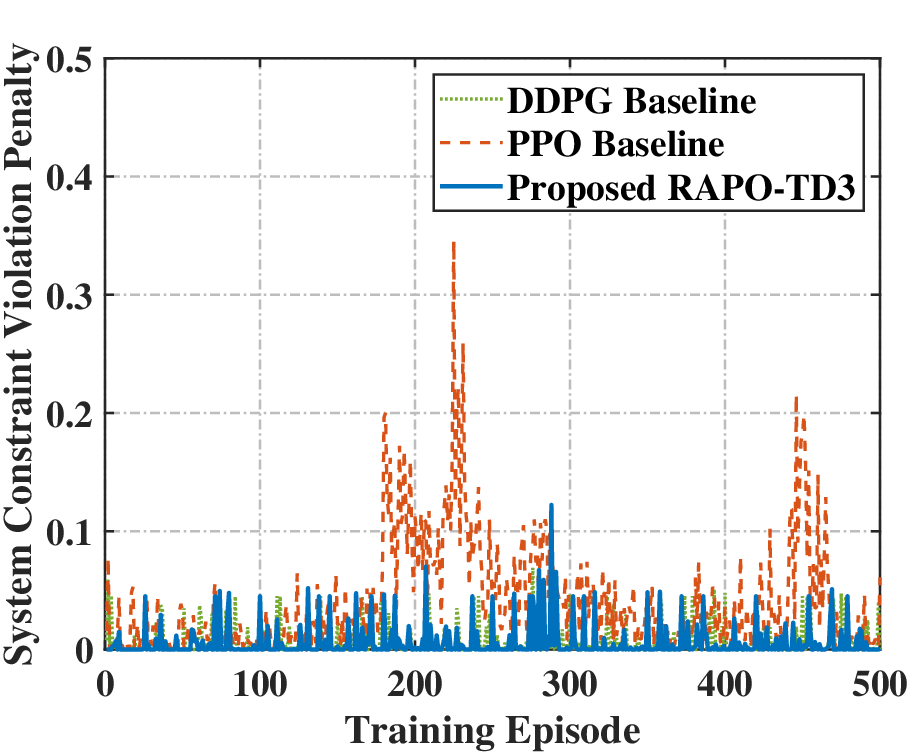}
    }\\
    \subfigure[URLLC QoS Satisfaction Ratio]{
        \label{fig:convergence:c}
        \includegraphics[width=0.48\linewidth]{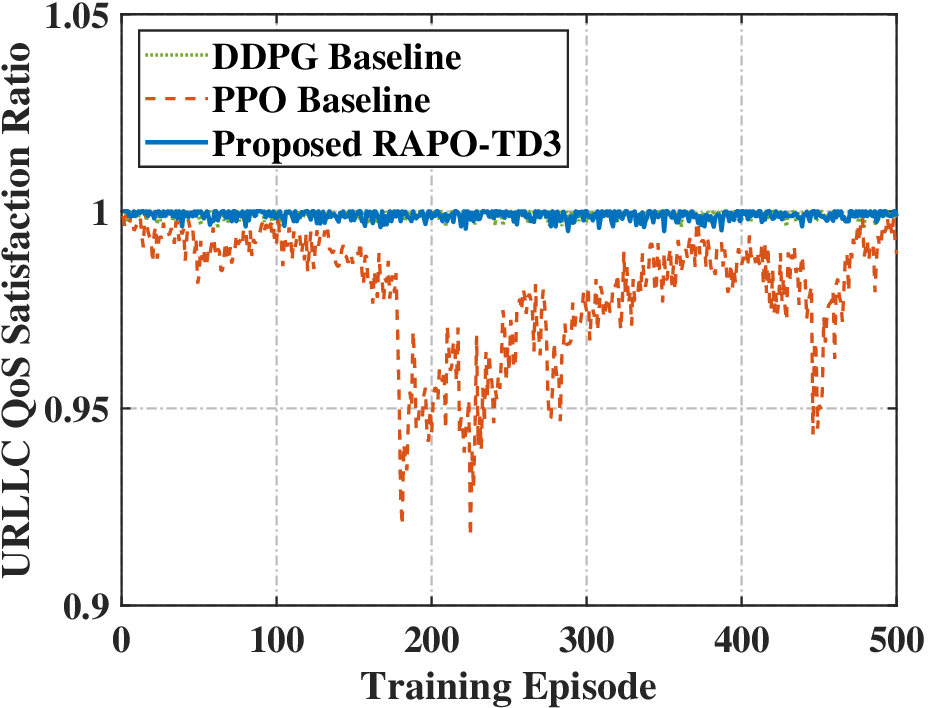}
    }%
    \subfigure[mMTC QoS Satisfaction Ratio]{
        \label{fig:convergence:d}
        \includegraphics[width=0.48\linewidth]{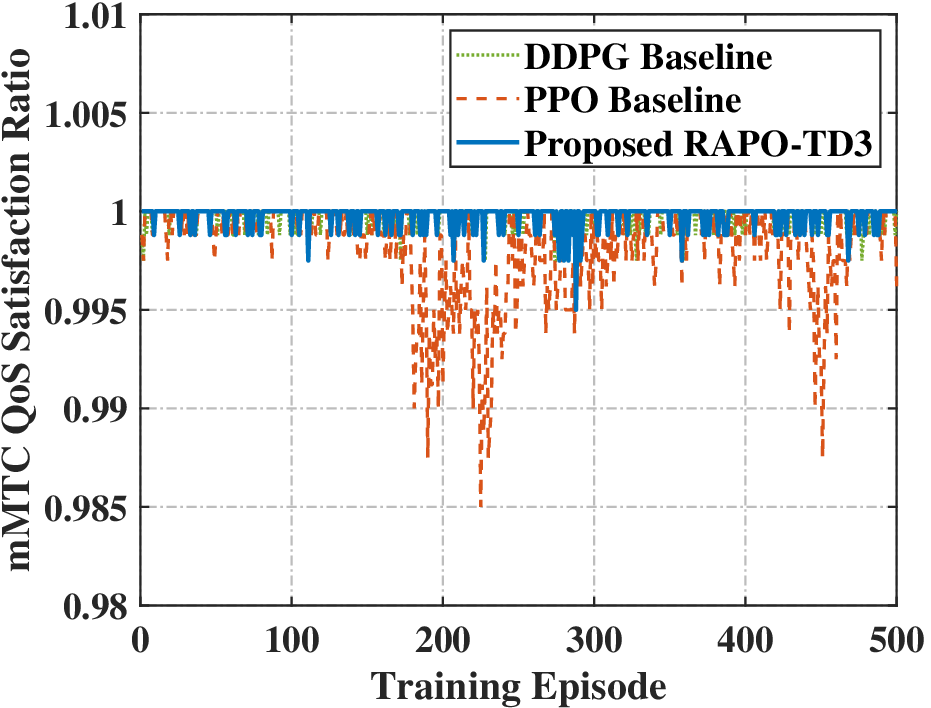}
    }
	\caption{Convergence performance and constraint satisfaction dynamics of different DRL algorithms under the normal 24-hour periodic traffic scenario: \subref{fig:convergence:a} Total cumulative reward, \subref{fig:convergence:b} Constraint violation penalty, \subref{fig:convergence:c} URLLC QoS satisfaction ratio, \subref{fig:convergence:d} mMTC QoS satisfaction ratio.}
    \label{fig:convergence}
\end{figure}

\subsection{Resilience and Resource Allocation Under Extreme URLLC Traffic Surges}
\label{sec:burst_performance}
To fully evaluate the operational resilience and stress-tolerance limits of the proposed framework during mission-critical emergencies, we expose the trained policies to a non-stationary bursty traffic profile. This scenario introduces acute $500\%$ surges in URLLC packet arrival rates, serving as a rigorous benchmark to test the dynamic priority orchestration agility of the MSO. The transient tracking of multi-service satisfaction ratios, constraint penalties, and resource allocations over 500 episodes is consolidated in Fig.~\ref{fig:resilience}.

As illustrated in Figs.~\ref{fig:resilience:a} and \ref{fig:resilience:b}, the PPO baseline exposes the structural vulnerability of unconstrained reward-maximization strategies. Driven by the impulse to maximize raw generalized utility, the PPO agent consistently over-allocates resources to eMBB services, maintaining an allocation profile near 0.98 as depicted in Fig.~\ref{fig:resilience:c}. When stochastic URLLC traffic spikes occur, this behavior triggers severe resource contention. Consequently, the PPO policy incurs sharp constraint violation penalties that spike above 3.0 and 3.3 around episodes 320 and 400, respectively. These boundary breaches cause the URLLC QoS satisfaction ratio to experience significant degradation, dropping to a low of 0.77. Such operational deficits indicate that standard continuous-control algorithms fail to safeguard mission-critical communication links under non-stationary pressure.

Conversely, the DDPG baseline exhibits an overly conservative learning trajectory. Although it appears to maintain acceptable URLLC satisfaction with a low penalty footprint (Figs.~\ref{fig:resilience:a} and \ref{fig:resilience:b}), its structural deficiency is clearly unveiled in Fig.~\ref{fig:resilience:c}. Due to severe policy underestimation bias, the DDPG agent fails to effectively optimize the multi-service action space, with its normalized eMBB resource allocation stagnating below 0.75 even after 500 episodes of training. This conservative suboptimality severely underutilizes the available aerial backhaul capacity, rendering it impractical for high-throughput cooperative networking.

\begin{figure}[!t]
    \centering
    \subfigure[URLLC QoS Satisfaction Ratio]{
        \label{fig:resilience:a}
        \includegraphics[width=0.48\linewidth]{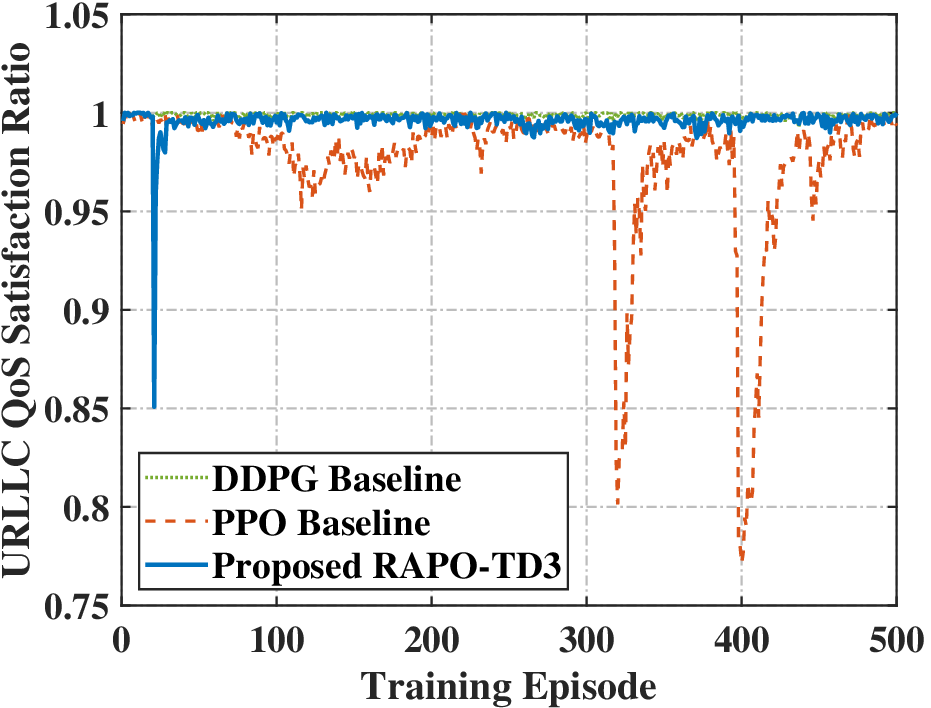}
    }%
    \subfigure[Constraint Violation Penalty]{
        \label{fig:resilience:b}
        \includegraphics[width=0.48\linewidth]{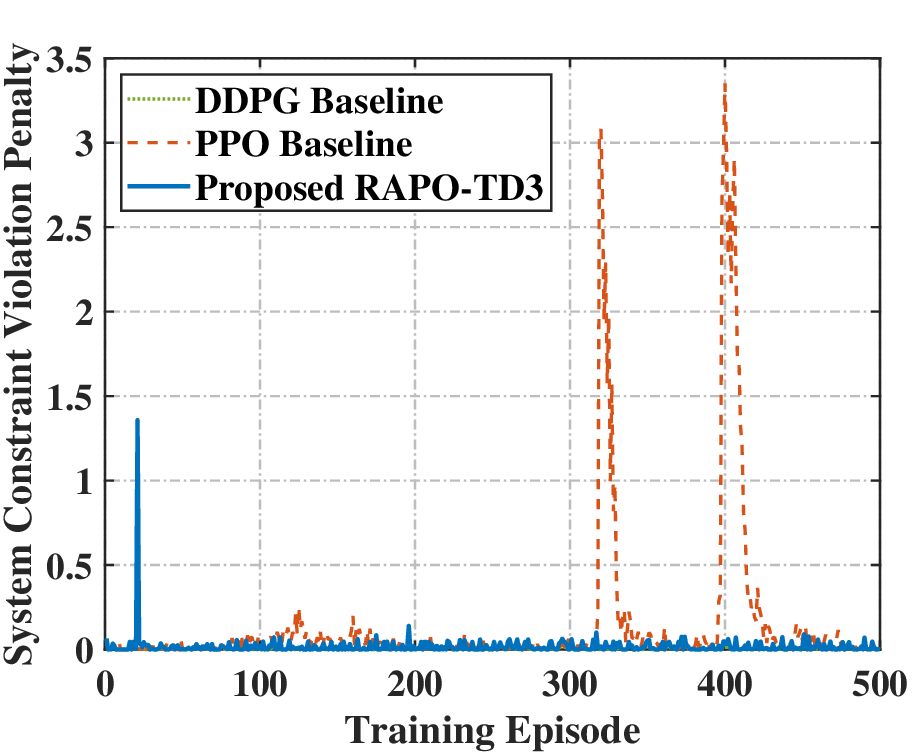}
    }\\
    \subfigure[eMBB QoS Satisfaction Ratio]{
        \label{fig:resilience:c}
        \includegraphics[width=0.48\linewidth]{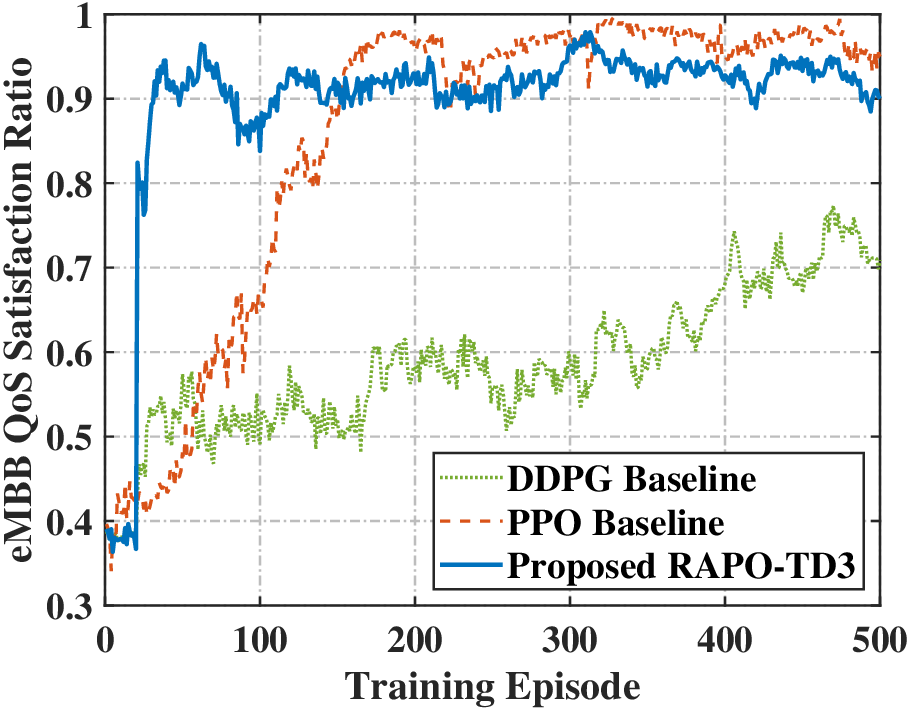}
    }%
    \subfigure[mMTC QoS Satisfaction Ratio]{
        \label{fig:resilience:d}
        \includegraphics[width=0.48\linewidth]{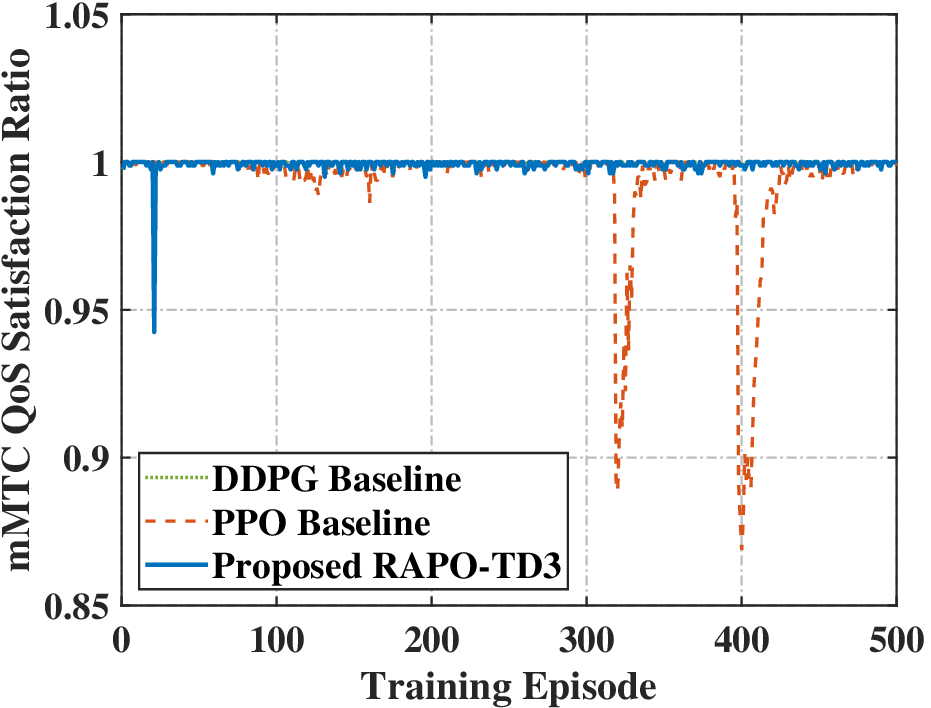}
    }
	\caption{Dynamic resilience and resource allocation performance of different DRL algorithms under the extreme URLLC traffic surge scenario: \subref{fig:resilience:a} URLLC QoS satisfaction ratio, \subref{fig:resilience:b} Constraint violation penalty, \subref{fig:resilience:c} eMBB QoS satisfaction ratio, and \subref{fig:resilience:d} mMTC QoS satisfaction ratio.}
    \label{fig:resilience}
\end{figure}

In sharp contrast, the proposed RAPO-TD3 framework armed with the RAPO mechanism learns a constraint-aware high-performing policy. As illustrated in Figs.~\ref{fig:resilience:a} and \ref{fig:resilience:b}, the proposed framework undergoes a transient penalty spike reaching 1.35 and a corresponding satisfaction dip to 0.85 near episode 25, reflecting initial structural adjustments under extreme environmental volatility. Following this brief exploratory adaptation, the system rapidly stabilizes. 

Driven by the dynamic weight vector $\boldsymbol{\omega}(t)$ generated by the adaptive smooth gating structure, the agent seamlessly scales down best-effort mMTC resource shares to accommodate the emergency surges. Beyond episode 50, the proposed RAPO-TD3 framework maintains a stable eMBB allocation above 0.93 while simultaneously shielding the latency-sensitive URLLC slice, keeping its satisfaction ratio tightly hovering at the near-optimal 1.0 boundary with near-zero system penalties. This resource allocation elasticity is comprehensively validated in Fig.~\ref{fig:resilience:d}, which tracks the dynamic trajectory of the mMTC QoS satisfaction ratio. The RAPO-TD3 framework suffers only an isolated transient dip to 0.942 at episode 22 during the initial training phase, after which it rapidly recalibrates and preserves the near-optimal 1.0 satisfaction boundary. Meanwhile, the unconstrained PPO baseline undergoes pronounced service drops decreasing sharply to 0.89 and 0.87 near episodes 320 and 400, and the DDPG baseline remains locked in a highly inefficient and conservative operational zone. Our framework is able to circumvent these suboptimalities, demonstrating how the RAPO online tuning loop effectively secures dual-layer isolation under extreme 6G backhaul anomalies.

\subsection{Ablation Study and Architectural Efficacy}
\label{sec:ablation_study}
To isolate and quantify the specific performance gains contributed by the twin-critic architecture and validate the necessity of mitigating overestimation bias in volatile aerial environments, we conduct an \textit{ablation study}. We introduce the 1Q-Optimized baseline, a structural variant that degrades the RAPO-TD3 framework by utilizing only a single critic network and removing the delayed policy smoothing mechanisms, while preserving identical state-action spaces and prioritized replay configurations. The comparative training dynamics over 500 episodes under the extreme traffic surge scenario are illustrated in Fig.~\ref{fig:ablation}.

A rigorous analysis of the total cumulative reward in Fig.~\ref{fig:ablation:a} demonstrates the superior learning efficiency of the proposed RAPO-TD3 framework. Armed with the twin-critic mechanism, the proposed framework achieves rapid convergence, reaching a high-utility stable plateau of approximately 2.95 within the first 50 episodes and maintaining consistent policy stability throughout the remaining training duration. Conversely, the 1Q-Optimized baseline exhibits a severely impeded learning trajectory, requiring nearly 350 episodes to approach a sub-optimal reward of 2.8. More critically, around episode 400, the 1Q-Optimized baseline undergoes a significant performance degradation, with its cumulative reward declining toward 2.65. This instability provides empirical evidence of unmitigated overestimation bias, where a single critic network consistently overvalues poor resource-slicing actions, leading to propagated gradient errors that eventually destabilize the actor policy.

The operational implications of this architectural deficiency are further illuminated by the system constraint violation penalty depicted in Fig.~\ref{fig:ablation:b}. Following the initial exploration phase, both the proposed framework and the 1Q-Optimized variant successfully suppress large-scale boundary breaches. However, a deeper examination of the steady-state volatility inset (spanning episodes 450 to 500) reveals distinct behavioral patterns. The proposed RAPO-TD3 framework exhibits minor boundary-tracking oscillations that peak boundedly near 0.08, whereas the 1Q-Optimized baseline maintains a lower penalty ripple. When evaluated jointly with the reward profiles from Fig.~\ref{fig:ablation:a}, this lower penalty profile indicates that the 1Q-Optimized agent has trapped itself within an overly conservative, suboptimal operational region that fails to fully exploit the available backhaul bandwidth capacity. In contrast, the proposed twin-critic design provides highly accurate value estimations, empowering the MSO to effectively adapt to and utilize the physical resource capacity frontier to maximize service utility while keeping localized SLA violations strictly constrained within a safe, negligible footprint.

\begin{figure}[!t]
    \centering
    \subfigure[Total Cumulative Reward]{
        \label{fig:ablation:a}
        \includegraphics[width=0.48\linewidth]{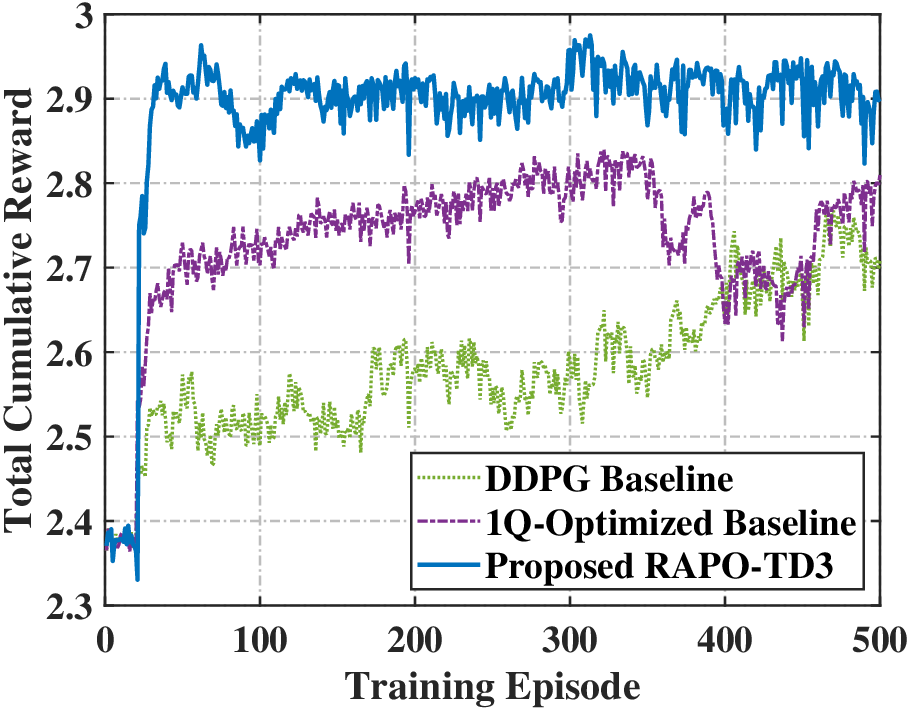}
    }%
    \subfigure[Constraint Violation Penalty]{
        \label{fig:ablation:b}
        \includegraphics[width=0.48\linewidth]{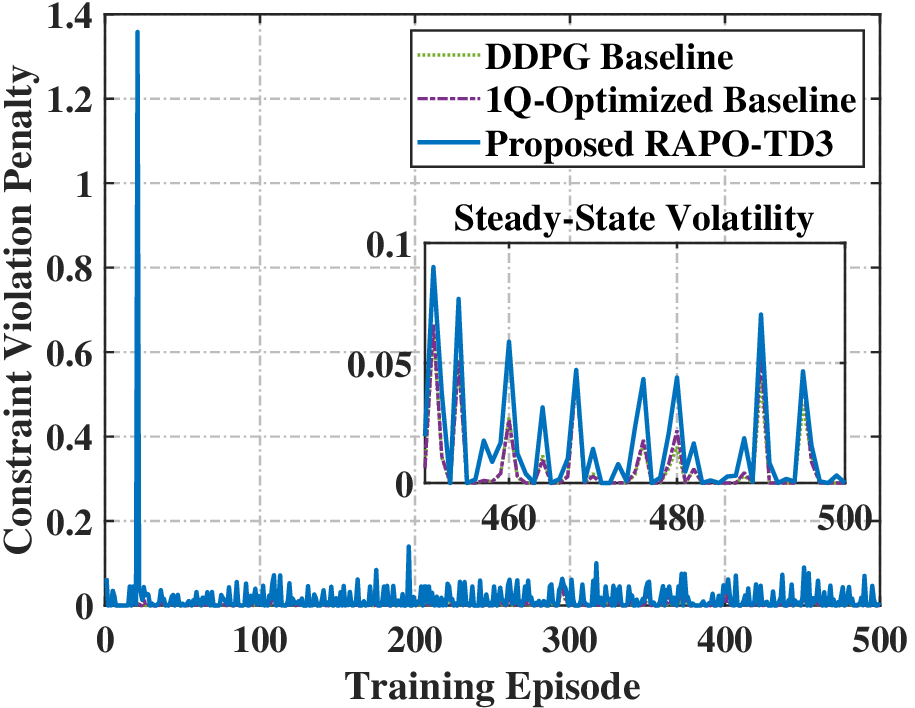}
    }
    \caption{Ablation study evaluating the structural efficacy of the twin-critic design under the extreme traffic scenario: \subref{fig:ablation:a} Total cumulative reward dynamics, and \subref{fig:ablation:b} Constraint violation penalty with a steady-state volatility inset.}
    \label{fig:ablation}
\end{figure}

\subsection{Scalability Analysis and Real-Time Execution Feasibility}

To evaluate the scalability and practical deployment feasibility of the proposed RAPO-TD3 framework, we conduct an additional experiment under varying network scales ($N \in \{2, 4, 6, 8\}$). The objective is to verify whether the computational efficiency and convergence stability can be maintained as the dense interference cluster expands.

As demonstrated in Table~\ref{tab:scalability}, the proposed architecture is highly scalable and yields two critical insights. First, while the framework maintains stable convergence across all scales, the steady-state average reward exhibits a marginal decline as $N$ increases. This phenomenon is fundamentally driven by physical-layer realities: a higher number of HAAs within a single frequency-reuse cluster intensifies co-channel interference, inherently degrading the overall Signal-to-Interference-plus-Noise Ratio (SINR) and tightening the feasible multi-objective allocation space. Second, the inference latency experiences a slight, strictly linear increase from $0.1260$ ms to $0.1442$ ms. This is attributed to the $\mathcal{O}(N)$ expansion of the state-action dimensionality, which linearly increases the floating-point operations in the neural network's input and output layers. 

Importantly, this computational overhead consumes only a negligible fraction of the stringent $1$ ms URLLC latency budget, leaving ample time for physical-layer transmission and queuing. These results theoretically and empirically suggest that the primary bottleneck in massive HAA deployments is physical interference rather than algorithmic scalability, indicating the stable applicability of our framework.

\section{Conclusion}
\label{sec:conclusion}

In this paper, we have proposed a stable dual-layer network slicing framework for HAA-assisted backhaul networks using a RAPO-TD3 deep reinforcement learning approach. To orchestrate heterogeneous service requirements under non-stationary traffic, we introduced a double soft-max projection mechanism and the RAPO scheme to guarantee physical constraint compliance and URLLC resilience. Furthermore, we presented a theoretical analysis showing that the proposed architecture satisfies the Robbins-Monro conditions and ensures Lipschitz continuity, supporting asymptotic convergence under the stated assumptions. Extensive simulations and scalability analyses under varying network dimensions demonstrate that our framework achieves better performance than existing baselines and achieves sub-millisecond inference latency. This computational efficiency consumes only a negligible fraction of the 1 ms URLLC latency budget while maintaining stable steady-state rewards against physical co-channel interference. Future work will extend this framework to accommodate unanchored small-scale UAVs by jointly optimizing 3D spatial trajectory mechanics, energy-aware hovering dynamics, and end-to-end fronthaul-backhaul slicing via decentralized federated learning paradigms for 6G coordination.

\begin{table}[!t]
\centering
\caption{Scalability Analysis of the Proposed Framework}
\label{tab:scalability}
\begin{tabular}{ccc}
\hline
HAAs ($N$) & Avg. Steady-State Reward & Inference Time (ms) \\
\hline
2 & 2.9495 & 0.1260 \\
4 & 2.8795 & 0.1326 \\
6 & 2.7922 & 0.1365 \\
8 & 2.5556 & 0.1442 \\
\hline
\end{tabular}
\end{table}


\color{black}




\bibliographystyle{IEEEtran}
\bibliography{IEEEabrv,reference}

@STRING{IEEE_J_ITS        = "{IEEE} Trans. Intell. Transp. Syst."}

@STRING{IEEE_J_VT         = "{IEEE} Trans. Veh. Technol."}

@STRING{IEEE_J_JSAC       = "{IEEE} J. Sel. Areas Commun."}

@STRING{IEEE_J_COM        = "{IEEE} Trans. Commun."}

@STRING{IEEE_J_WCOM       = "{IEEE} Trans. Wireless Commun."}

@STRING{IEEE_J_WCOML      = "{IEEE} Wireless Commun. Lett."}

@STRING{IEEE_J_GCN        = "{IEEE} Trans. Green Commun. Netw."}

@STRING{IEEE_J_IOT        = "{IEEE} Internet Things J."}

@STRING{IEEE_J_MC         = "{IEEE} Trans. Mobile Comput."}

@STRING{IEEE_J_NSE        = "{IEEE} Trans. Netw. Sci. Eng."}

@STRING{IEEE_J_NSM        = "{IEEE} Trans. Netw. Service Manag."}

@STRING{IEEE_J_CCN        = "{IEEE} Trans. on Cogn. Commun. Netw."}

@STRING{IEEE_J_PROC       = "Proc. {IEEE}"}

@STRING{IEEE_O_ACC        = "{IEEE} Access"}

@STRING{IEEE_M_COM        = "{IEEE} Commun. Mag."}

@STRING{IEEE_O_CSTO       = "{IEEE} Commun. Surveys Tuts."}

@STRING{IEEE_M_SP         = "{IEEE} Signal Process. Mag."}

@STRING{IEEE_M_VT         = "{IEEE} Veh. Technol. Mag."}

@STRING{IEEE_M_WC         = "{IEEE} Wireless Commun."}

@article{8685766,
  author  = {Zhang, Shunliang},
  journal = IEEE_M_WC,
  title   = {An Overview of Network Slicing for {5G}},
  year    = {2019},
  volume  = {26},
  number  = {3},
  month   = jun,
  pages   = {111-117},
  doi     = {10.1109/MWC.2019.1800234}
}

@article{Xiao2018,
  author  = {Xiao, Yong and Hirzallah, Mohammed and Krunz, Marwan},
  journal = IEEE_J_JSAC,
  title   = {Distributed Resource Allocation for Network Slicing Over Licensed and Unlicensed Bands},
  year    = {2018},
  volume  = {36},
  number  = {10},
  month   = oct,
  pages   = {2260-2274},
  doi     = {10.1109/JSAC.2018.2869964}
}

@article{9729992,
  author  = {Filali, Abderrahime and Mlika, Zoubeir and Cherkaoui, Soumaya and Kobbane, Abdellatif},
  journal = IEEE_J_NSE,
  title   = {Dynamic {SDN}-Based Radio Access Network Slicing With Deep Reinforcement Learning for {URLLC} and {eMBB} Services},
  year    = {2022},
  volume  = {9},
  number  = {4},
  month   = {July-Aug.},
  pages   = {2174-2187},
  doi     = {10.1109/TNSE.2022.3157274}
}

@article{10186380,
  author  = {Cao, Haotong and Kumar, Neeraj and Yang, Longxiang and Guizani, Mohsen and Yu, F. Richard},
  journal = IEEE_J_NSM,
  title   = {Resource Orchestration and Allocation of {E2E} Slices in Softwarized {UAVs}-Assisted {6G} Terrestrial Networks},
  year    = {2024},
  volume  = {21},
  number  = {1},
  month   = feb,
  pages   = {1032-1047},
  doi     = {10.1109/TNSM.2023.3296858}
}

@article{10003191,
  author  = {Peng, Haoran and Wang, Li-Chun and Jian, Zhuofu},
  journal = IEEE_J_CCN,
  title   = {Data-Driven Spectrum Partition for Multiplexing {URLLC} and {eMBB}},
  year    = {2023},
  volume  = {9},
  number  = {2},
  pages   = {386-397},
  month   = apr,
  doi     = {10.1109/TCCN.2022.3231690}
}

@article{10323247,
  author  = {Tian, Mengqiu and Li, Changle and Hui, Yilong and Cheng, Nan and Yue, Wenwei and Fu, Yuchuan and Han, Zhu},
  journal = IEEE_J_ITS,
  title   = {On-Demand Multiplexing of {eMBB}/{URLLC} Traffic in a Multi-{UAV} Relay Network},
  year    = {2024},
  volume  = {25},
  number  = {6},
  pages   = {6035-6048},
  month   = jun,
  doi     = {10.1109/TITS.2023.3332022}
}

@article{10354514,
  author  = {Taskou, Shiva Kazemi and Rasti, Mehdi and Hossain, Ekram},
  journal = IEEE_J_MC,
  title   = {End-to-End Resource Slicing for Coexistence of {eMBB} and {URLLC} Services in {5G}-Advanced/{6G} Networks},
  year    = {2024},
  volume  = {23},
  number  = {7},
  pages   = {8015-8032},
  month   = jul,
  doi     = {10.1109/TMC.2023.3341810}
}

@article{10742112,
  author  = {Tul Muntaha, Sidra and Hafeez, Maryam and Ahmed, Qasim Z. and Khan, Faheem A. and Zaharis, Zaharias D. and Lazaridis, Pavlos I.},
  journal = IEEE_J_CCN,
  title   = {{RAN} Slicing With Joint Resource Allocation for a Multi-Tenant–Multi-Service System},
  year    = {2025},
  volume  = {11},
  number  = {3},
  pages   = {1927-1939},
  month   = jun,
  doi     = {10.1109/TCCN.2024.3490781}
}

@inproceedings{9448942,
  author    = {Tominaga, Eduardo Noboro and Alves, Hirley and Souza, Richard Demo and Luiz Rebelatto, João and Latva-aho, Matti},
  booktitle = {IEEE 93rd Vehicular Technology Conference (VTC2021-Spring)},
  title     = {Non-Orthogonal Multiple Access and Network Slicing: Scalable Coexistence of {eMBB} and {URLLC}},
  year      = {2021},
  address   = {Helsinki, Finland},
  doi       = {10.1109/VTC2021-Spring51267.2021.9448942}
}

@article{9232929,
  author  = {Zhou, Fanqin and Yu, Peng and Feng, Lei and Qiu, Xuesong and Wang, Zhili and Meng, Luoming and Kadoch, Michel and Gong, Liang and Yao, Xianjiong},
  journal = IEEE_M_WC,
  title   = {Automatic Network Slicing for {IoT} in Smart City},
  year    = {2020},
  volume  = {27},
  number  = {6},
  pages   = {108-115},
  month   = dec,
  doi     = {10.1109/MWC.001.2000069}
}

@article{9975290,
  author  = {Coelho, André and Rodrigues, João and Fontes, Helder and Campos, Rui and Ricardo, Manuel},
  journal = IEEE_O_ACC,
  title   = {An Algorithm for Placing and Allocating Communications Resources Based on Slicing-Aware Flying Access and Backhaul Networks},
  year    = {2022},
  volume  = {10},
  number  = {},
  month   = dec,
  pages   = {128923-128942},
  doi     = {10.1109/ACCESS.2022.3227653}
}

@article{10473184,
  author  = {Guan, Yingying and Song, Qingyang and Chen, Tao and Qi, Weijing and Guo, Lei and Jamalipour, Abbas},
  journal = IEEE_M_VT,
  title   = {Slicing-Aware Aerial Networks for Integrated Sensing and Communication: {3D} Placement and Adaptive Allocation of Resources},
  year    = {2024},
  volume  = {19},
  number  = {2},
  pages   = {79-88},
  month   = jun,
  doi     = {10.1109/MVT.2024.3372509}
}

@article{10551400,
  author  = {Rafique, Wajid and Rani Barai, Joyeeta and Fapojuwo, Abraham O. and Krishnamurthy, Diwakar},
  journal = IEEE_O_CSTO,
  title   = {A Survey on Beyond {5G} Network Slicing for Smart Cities Applications},
  year    = {2025},
  volume  = {27},
  number  = {1},
  pages   = {595-628},
  month   = feb,
  doi     = {10.1109/COMST.2024.3410295}
}

@article{10104142,
  author  = {Liu, Yalin and Wang, Qiu and Dai, Hong-Ning and Fu, Yaru and Zhang, Ning and Lee, Chi Chung},
  journal = IEEE_J_VT,
  title   = {{UAV}-Assisted Wireless Backhaul Networks: Connectivity Analysis of Uplink Transmissions},
  year    = {2023},
  volume  = {72},
  number  = {9},
  pages   = {12195-12207},
  month   = sep,
  doi     = {10.1109/TVT.2023.3268025}
}

@article{s24061888,
  author         = {Chataut, Robin and Nankya, Mary and Akl, Robert},
  title          = {{6G} Networks and the {AI} Revolution—Exploring Technologies, Applications, and Emerging Challenges},
  journal        = {Sensors},
  volume         = {24},
  year           = {2024},
  number         = {6},
  article-number = {1888},
  month          = mar,
  doi            = {10.3390/s24061888}
}

@article{YU2023109533,
  title   = {Coordinated parallel resource allocation for integrated access and backhaul networks},
  journal = {Computer Networks},
  volume  = {222},
  pages   = {109533},
  year    = {2023},
  issn    = {1389-1286},
  month   = feb,
  doi     = {10.1016/j.comnet.2022.109533},
  author  = {Mengxin Yu and Yibo Pi and Aimin Tang and Xudong Wang}
}

@article{10716774,
  author  = {Ataeebojd, Elaheh and Rasti, Mehdi and Latva-Aho, Matti},
  journal = IEEE_J_GCN,
  title   = {Network Selection and Resource Allocation for Coexistence of {eMBB} and {URLLC} Services in a {6G} Multi-Band {HetNet}},
  year    = {2025},
  volume  = {9},
  number  = {3},
  pages   = {1179-1194},
  month   = sep,
  doi     = {10.1109/TGCN.2024.3481281}
}

@article{10537066,
  author  = {Tian, Mengqiu and Li, Changle and Hui, Yilong and Chen, Binbin and Yue, Wenwei and Fu, Yuchuan and Han, Zhu},
  journal = IEEE_J_WCOM,
  title   = {An Intelligent Coexistence Strategy for {eMBB}/{URLLC} Traffic in Multi-{UAV} Relay Networks via Deep Reinforcement Learning},
  year    = {2024},
  volume  = {23},
  number  = {10},
  pages   = {13424-13439},
  month   = oct,
  doi     = {10.1109/TWC.2024.3401163}
}

@article{9999295,
  author  = {Abouaomar, Amine and Taik, Afaf and Filali, Abderrahime and Cherkaoui, Soumaya},
  journal = IEEE_M_COM,
  title   = {Federated Deep Reinforcement Learning for Open {RAN} Slicing in {6G} Networks},
  year    = {2023},
  volume  = {61},
  number  = {2},
  pages   = {126-132},
  month   = feb,
  doi     = {10.1109/MCOM.007.2200555}
}

@article{8103164,
  author  = {Arulkumaran, Kai and Deisenroth, Marc Peter and Brundage, Miles and Bharath, Anil Anthony},
  journal = IEEE_M_SP,
  title   = {Deep Reinforcement Learning: A Brief Survey},
  year    = {2017},
  volume  = {34},
  number  = {6},
  pages   = {26-38},
  month   = nov,
  doi     = {10.1109/MSP.2017.2743240}
}

@article{LADOSZ20221,
  title   = {Exploration in deep reinforcement learning: A survey},
  journal = {Information Fusion},
  volume  = {85},
  pages   = {1-22},
  year    = {2022},
  issn    = {1566-2535},
  month   = sep,
  doi     = {10.1016/j.inffus.2022.03.003},
  author  = {Pawel Ladosz and Lilian Weng and Minwoo Kim and Hyondong Oh}
}

@article{10302364,
  author  = {Cai, Yue and Cheng, Peng and Chen, Zhuo and Ding, Ming and Vucetic, Branka and Li, Yonghui},
  journal = IEEE_J_MC,
  title   = {Deep Reinforcement Learning for Online Resource Allocation in Network Slicing},
  year    = {2024},
  volume  = {23},
  number  = {6},
  pages   = {7099-7116},
  month   = jun,
  doi     = {10.1109/TMC.2023.3328950}
}

@article{10750921,
  author  = {Chen, Geng and Mu, Xinzheng and Liang, Hongjia and Zeng, Qingtian and Zhang, Yu-Dong},
  journal = IEEE_J_WCOM,
  title   = {Distributed {RAN} Slicing Based on {MATD3} Joint With Evolutionary Game Assisted User Association for {MEC}-Enabled {HetNets}},
  year    = {2025},
  volume  = {24},
  number  = {1},
  pages   = {260-276},
  month   = jan,
  doi     = {10.1109/TWC.2024.3491359}
}

@article{s22083031,
  author         = {Hurtado Sánchez, Johanna Andrea and Casilimas, Katherine and Caicedo Rendon, Oscar Mauricio},
  title          = {Deep Reinforcement Learning for Resource Management on Network Slicing: A Survey},
  journal        = {Sensors},
  volume         = {22},
  year           = {2022},
  number         = {8},
  article-number = {3031},
  month          = apr,
  doi            = {10.3390/s22083031}
}

@article{9459763,
  author  = {Mei, Jie and Wang, Xianbin and Zheng, Kan and Boudreau, Gary and Sediq, Akram Bin and Abou-Zeid, Hatem},
  journal = IEEE_J_COM,
  title   = {Intelligent Radio Access Network Slicing for Service Provisioning in {6G}: A Hierarchical Deep Reinforcement Learning Approach},
  year    = {2021},
  volume  = {69},
  number  = {9},
  pages   = {6063-6078},
  month   = sep,
  doi     = {10.1109/TCOMM.2021.3090423}
}

@article{10026882,
  author  = {Kang, Hongyue and Chang, Xiaolin and Mišić, Jelena and Mišić, Vojislav B. and Fan, Junchao and Liu, Yating},
  journal = IEEE_J_IOT,
  title   = {Cooperative {UAV} Resource Allocation and Task Offloading in Hierarchical Aerial Computing Systems: A {MAPPO}-Based Approach},
  year    = {2023},
  volume  = {10},
  number  = {12},
  pages   = {10497-10509},
  month   = jun,
  doi     = {10.1109/JIOT.2023.3240173}
}

@article{9946428,
  author  = {Lai, Chuan-Chi and Bhola and Tsai, Ang-Hsun and Wang, Li-Chun},
  journal = IEEE_J_VT,
  title   = {Adaptive and Fair Deployment Approach to Balance Offload Traffic in Multi-{UAV} Cellular Networks},
  year    = {2023},
  volume  = {72},
  number  = {3},
  pages   = {3724-3738},
  month   = mar,
  doi     = {10.1109/TVT.2022.3221557}
}

@inproceedings{10257231,
  author    = {Bellone, Lorenzo and Galkin, Boris and Traversi, Emiliano and Natalizio, Enrico},
  booktitle = {The 19th International Conference on Distributed Computing in Smart Systems and the {Internet of Things} (DCOSS-IoT)},
  title     = {Deep Reinforcement Learning for Combined Coverage and Resource Allocation in {UAV}-Aided {RAN}-Slicing},
  year      = {2023},
  volume    = {},
  number    = {},
  address   = {Pafos, Cyprus},
  doi       = {10.1109/DCOSS-IoT58021.2023.00106}
}

@techreport{3gpp_tr_36_777,
  author      = {{3rd Generation Partnership Project (3GPP)}},
  title       = {{Study on Enhanced LTE Support for Aerial Vehicles}},
  institution = {3rd Generation Partnership Project (3GPP)},
  type        = {Tech. Rep.},
  number      = {36.777},
  year        = {2017}
}

@techreport{3gpp_tr_38_901,
  author      = {{3rd Generation Partnership Project (3GPP)}},
  title       = {Study on channel model for frequencies from 0.5 to 100 GHz (release 14)},
  institution = {3rd Generation Partnership Project (3GPP)},
  type        = {Tech. Rep.},
  number      = {38.901 V16.1.0},
  year        = {2017}
}

@inproceedings{PER_2016,
  author    = {Tom Schaul and
               John Quan and
               Ioannis Antonoglou and
               David Silver},
  title     = {Prioritized Experience Replay},
  booktitle = {The 4th International Conference on Learning Representations ({ICLR})},
  address   = {San Juan, Puerto Rico},
  year      = {2016}
}

@inproceedings{FujimotoHM18,
  author    = {Scott Fujimoto and
               Herke van Hoof and
               David Meger},
  title     = {Addressing Function Approximation Error in Actor-Critic Methods},
  booktitle = {The 35th International Conference on Machine Learning ({ICML})},
  address   = {Stockholmsm{\"{a}}ssan, Stockholm},
  year      = {2018}
}

@techreport{3gpp_ts_28_533,
  author      = {{3rd Generation Partnership Project (3GPP)}},
  title       = {Management and orchestration; Architecture framework},
  institution = {3rd Generation Partnership Project (3GPP)},
  type        = {Tech. Spec.},
  number      = {28.533},
  version     = {18.4.0},
  year        = {2024}
}

@article{al2014optimal,
  author  = {Al-Hourani, Akram and Kandeepan, Sithamparanathan and Lardner, Simon},
  journal = IEEE_J_WCOML,
  title   = {Optimal {LAP} Altitude for Maximum Coverage},
  year    = {2014},
  volume  = {3},
  number  = {6},
  pages   = {569-572},
  month   = dec,
  doi     = {10.1109/LWC.2014.2342736}
}

@article{khawaja2019survey,
  author  = {Khawaja, Wahab and Guvenc, Ismail and Matolak, David W. and Fiebig, Uwe-Carsten and Schneckenburger, Nicolas},
  journal = IEEE_O_CSTO,
  title   = {A Survey of Air-to-Ground Propagation Channel Modeling for Unmanned Aerial Vehicles},
  year    = {2019},
  volume  = {21},
  number  = {3},
  pages   = {2361-2391},
  month   = {thirdquarter},
  doi     = {10.1109/COMST.2019.2915069}
}

@article{Song2026,
  author  = {Song, Xin and Zhang, Biao and Fan, Ze and Li, Ruomeng and Xu, Siyang},
  journal = IEEE_J_NSM,
  title   = {Dynamic Normalization {TD3}-Based Task Offloading for {UAV}-Assisted Collaborative Computing},
  year    = {2026},
  volume  = {23},
  number  = {},
  pages   = {2762-2777},
  doi     = {10.1109/TNSM.2026.3667404}
}

@article{Ganjalizadeh2026,
  author  = {Ganjalizadeh, Milad and Ghadikolaei, Hossein S. and Gündüz, Deniz and Petrova, Marina},
  journal = IEEE_J_NSM,
  title   = {{BSAC-CoEx}: Coexistence of {URLLC} and Distributed Learning Services via Device Selection},
  year    = {2026},
  volume  = {23},
  number  = {},
  pages   = {1406-1421},
  doi     = {10.1109/TNSM.2025.3641848}
}

@article{Durisi2016,
  author  = {Durisi, Giuseppe and Koch, Tobias and Popovski, Petar},
  journal = IEEE_J_PROC,
  title   = {Toward Massive, Ultrareliable, and Low-Latency Wireless Communication With Short Packets},
  year    = {2016},
  volume  = {104},
  number  = {9},
  pages   = {1711-1726},
  month   = sep,
  doi     = {10.1109/JPROC.2016.2537298}
}

@article{She2018,
  author  = {She, Changyang and Yang, Chenyang and Quek, Tony Q. S.},
  journal = IEEE_J_WCOM,
  title   = {Cross-Layer Optimization for Ultra-Reliable and Low-Latency Radio Access Networks},
  year    = {2018},
  volume  = {17},
  number  = {1},
  pages   = {127-141},
  month   = jan,
  doi     = {10.1109/TWC.2017.2762684}
}

@book{Boyd2004,
  author    = {Boyd, Stephen and Vandenberghe, Lieven},
  title     = {Convex Optimization},
  publisher = {Cambridge University Press},
  year      = {2004}
}

@inproceedings{Silver2014,
  author    = {Silver, David and Lever, Guy and Heess, Nicolas and Degris, Thomas and Wierstra, Daan and Riedmiller, Martin},
  title     = {Deterministic policy gradient algorithms},
  year      = {2014},
  booktitle = {The 31st International Conference on Machine Learning ({ICML})},
  address   = {Beijing, China}
}

@article{Robbins1951,
  author  = {Robbins, Herbert and Monro, Sutton},
  journal = {The Annals of Mathematical Statistics},
  title   = {A Stochastic Approximation Method},
  year    = {1951},
  volume  = {22},
  number  = {3},
  pages   = {400-407}
}
\ifCLASSOPTIONcaptionsoff  \newpage \fi

\begin{IEEEbiography}[{\includegraphics[width=1in,height=1.25in,clip,keepaspectratio]{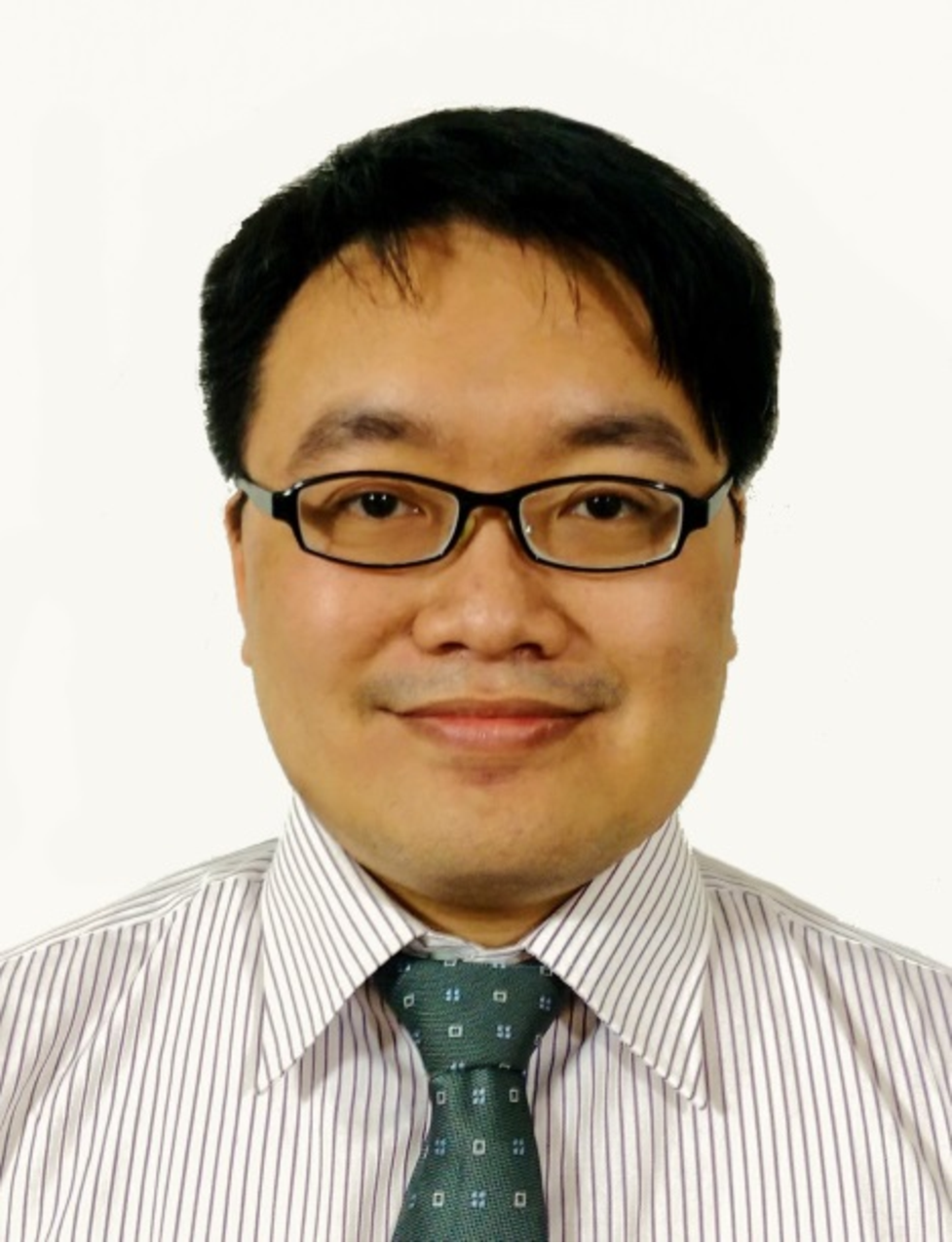}}]{Chuan-Chi Lai}
    (Member, IEEE) received the Ph.D. degree in Computer Science and Information Engineering from the National Taipei University of Technology, Taiwan, in 2017. He held research and faculty positions at National Chiao Tung University and Feng Chia University prior to his current role. Since 2024, he has been an Assistant Professor with the Department of Communications Engineering, National Chung Cheng University, Chiayi, Taiwan. His research interests include mobile edge computing, UAV networks, and AI for wireless communications. Dr. Lai was a recipient of the Postdoctoral Researcher Academic Research Award from the NSTC, Taiwan, in 2019, and Best Paper Awards at WOCC (2018, 2021) and ICUFN (2015).
\end{IEEEbiography}

\begin{IEEEbiography}[{\includegraphics[width=1in,height=1.25in,clip,keepaspectratio]{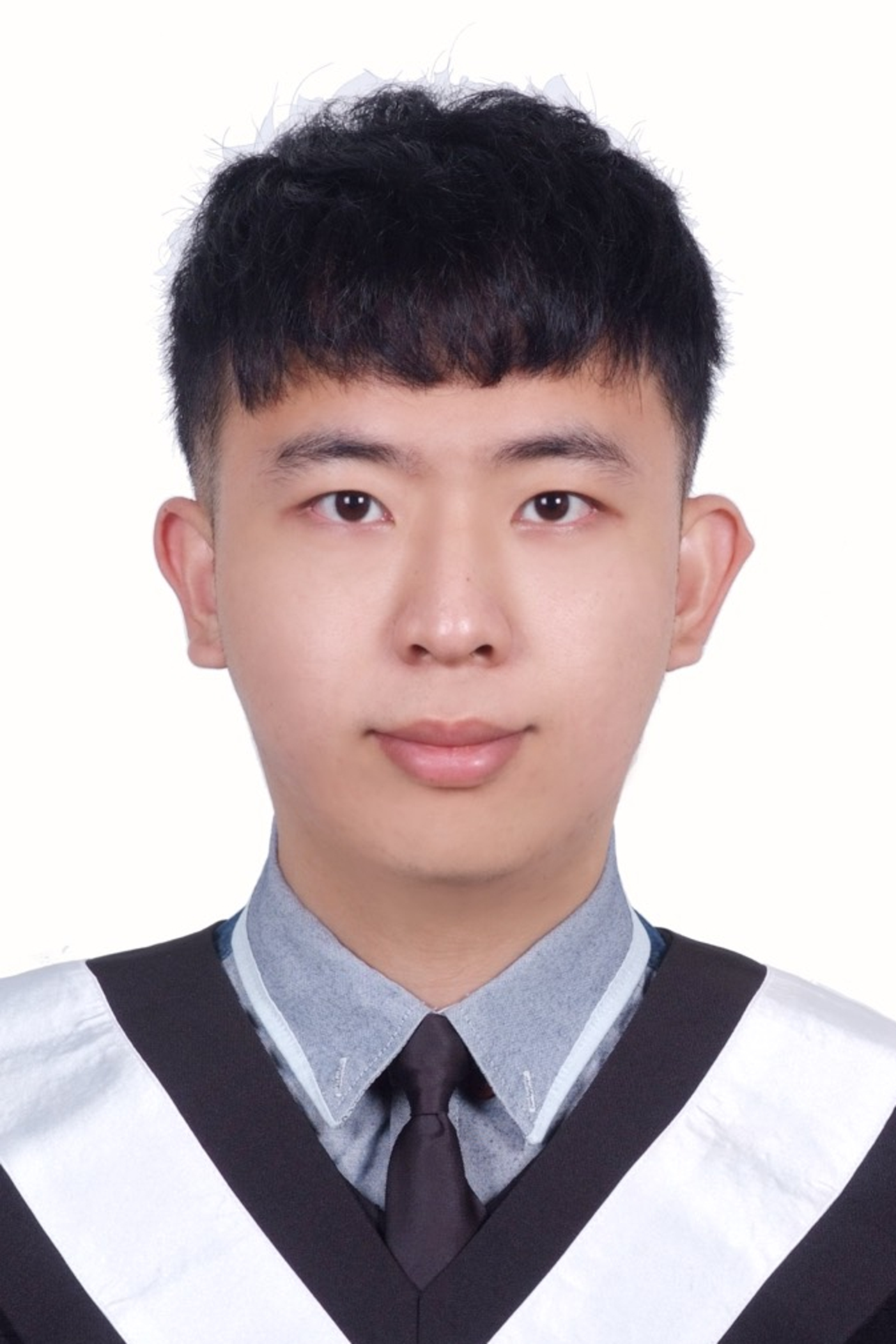}}]{Jen-Hsiang Li}
	received the bachelor's degree in Information Technology from Overseas Chinese University, Taichung, Taiwan, in 2023. In 2023, he joined the Applied Intelligent Communication and Computing Lab. Currently, he is pursuing his master's degree in Information Engineering and Computer Science at Feng Chia University, Taichung, Taiwan. His research interests include radio resource management, network slicing, and reinforcement learning. 
\end{IEEEbiography}

\end{document}